%% file: main.tex
\documentclass[11pt]{article}
\usepackage[utf8]{inputenc}
\usepackage[T1]{fontenc}
\usepackage{lmodern}
\usepackage{graphicx}  
\usepackage{amsmath, amsthm, amssymb}
\usepackage[colorlinks]{hyperref}
\usepackage[nameinlink,noabbrev,capitalise]{cleveref}
\usepackage{fullpage}
\usepackage{enumitem}
\usepackage[dvipsnames]{xcolor}
\usepackage{cancel}
\usepackage{color-edits}
\usepackage{physics}
\usepackage{thmtools}
\usepackage{thm-restate}

\usepackage{comment}

\usepackage{import}
\input{macros.tex}

\addauthor[Moshe]{mb}{red}
\addauthor[Yuval]{yg}{blue}
\addauthor[Noam]{nmm}{ForestGreen}

\hypersetup{
  colorlinks   = true,
  linkcolor    = blue,
  citecolor    = black,
  filecolor    = blue,
  urlcolor     = blue
}

\newtheorem{theorem}{Theorem}[section]
\newtheorem{lemma}[theorem]{Lemma}
\newtheorem{proposition}[theorem]{Proposition}
\newtheorem{corollary}[theorem]{Corollary}
\newtheorem{definition}[theorem]{Definition}

\newtheorem{example}[theorem]{Example}

\newcommand{\papertitle}{On Best-of-Both-Worlds Fairness\\ 
via Sum-of-Variances Minimization} 
\title{\papertitle}
\author{Moshe Babaioff\thanks{Hebrew University of Jerusalem (HUJI). Email: moshe.babaioff@mail.huji.ac.il}
\and Yuval Grofman\thanks{Hebrew University of Jerusalem (HUJI). Email: yuval.grofman.46@gmail.com}}
\date{\today}

\begin{document}
\maketitle

\begin{abstract}

We consider the problem of fairly allocating a set of indivisible goods among agents with additive valuations.  Ex-ante fairness (proportionality) can trivially be obtained by giving all goods to a random agent.  Yet, such an allocation is very unfair ex-post.  This has motivated the ``Best-of-Both-Worlds (BoBW)'' approach, seeking a randomized allocation that is ex-ante  proportional and is supported only on ex-post  fair allocations (e.g., on allocations that are envy-free-up-to-one-good (EF1), or give some constant fraction of the maximin share (MMS)). 
{It is commonly pointed out that the distribution  that allocates all goods to one agent at random fails  to be ex-post fair as it ignores the variances of the values of the agents.}
We examine the approach of trying to mitigate this 
problem by (explicitly) minimizing the sum-of-variances 
of the values of the agents, subject to ex-ante 
proportionality. We study the ex-post fairness properties 
of the resulting distributions.
In support of this approach, observe that such an 
optimization will indeed deterministically output 
a proportional allocation  if such exists. 
We show that when valuations are identical, this approach indeed guarantees fairness ex-post: all allocations in the support are envy-free-up-to-any-good (EFX), and thus guarantee every agent at least $4/7$ of her maximin share (but not her full MMS). On the negative side, 
we show that this approach completely fails 
when valuations are not  identical:  even in the simplest setting of only two agents and two goods, when the additive valuations are not identical, there is positive probability of allocating both goods to the same agent. Thus, the supporting ex-post allocation might not be EF1 (and thus also not EFX), and might not give an agent any constant fraction of her MMS. Finally, we present similar negative results for other natural minimization objectives that are based on variances or standard deviations.
\end{abstract}

\newpage
\input{intro}

\section{Preliminaries}\label{sec:preliminaries}

\subsection{The Model and Basic Notations}
We consider the allocation of a set $\mathcal{M}=\{1,2,\dots,m\}$ of $m$ indivisible goods to a set $\mathcal{N}=\{1,2,\ldots,n\}$ of $n$ agents. 
Each agent $i\in \mathcal{N}$ has an additive valuation function $v_i:2^{\mathcal{M}}\to\mathbb{R}_{\ge0}$, 
satisfying $v_i(S)=\sum_{g\in S} v_i(\{g\})$ for all $S\subseteq \mathcal{M}$ (and $v_i(\emptyset)=0$). 
We use $\vec{v}= (v_1,v_2, \ldots, v_n)$ to denote a vector of valuations for the $n$ agents, 
with agent $i\in \mathcal{N}$ having valuation $v_i$. 
We say that the \textit{valuations are normalized to $C \in \mathbb{R}_{\ge0}$}, if  
for all $i\in \mathcal{N}$ it holds that $v_i(\mathcal{M}) = C$.   
An \emph{allocation $(A_1,\dots,A_n)$ of $\mathcal{M}$ to $n$ agents}
is a partition of the set of goods $\mathcal{M}$  
into $n$ {disjoint} bundles 
($A_i \cap A_l = \emptyset$ for every  $i \neq l$),
{such that} 
no good remains unallocated ($\bigcup_{i\in\mathcal{N}} A_i = \mathcal{M}$). We let  $\Pi(n, \mathcal{M})$ denote the set of 
all allocations of the goods $\mathcal{M}$ to 
$n$ agents. 
For convenience, for any $K \in \mathbb{R}$ and $n$ agents, 
we denote the length $n$ vector $(K, K, \ldots, K) \in \mathbb{R}^n$ by $(K)^n$.
{Throughout the paper, we use $i$ to 
index an agent, and use $j$ to index 
an allocation in the support of a distribution.}
Finally, throughout the paper we only use the Euclidean (L2) norm, which we denote by $||\cdot||$.

\subsection{Fairness Criteria}

\begin{definition}[Proportional Share (PS)]
  For an agent with valuation $v$, her proportional share (PS) in an $n$-agents setting is defined to be
  \[ \textnormal{PS}(v, n, \mathcal{M}) = v(\mathcal{M}) / n. \]
  For a set $\mathcal{N}$ of agents 
  with vector of valuations $\vec{v}= (v_1,v_2, \ldots, v_n)$ over the set of goods $\mathcal{M}$,   
  we say that an allocation $(A_1,\dots,A_n)$ is \emph{proportional} 
  if $v_i(A_i) \ge \text{PS}(v_i, n, \mathcal{M})$ for every $i\in\mathcal{N}$. 
\end{definition}
  When valuations are normalized to $C$ (that is, $v_i(\mathcal{M})= C$ for every $i$), the proportional share of every agent is $C / n$, and we denote it by $ \text{PS}$.
We note that valuations are normalized in any 
setting with identical valuations (as the ones in 
\cref{sec:identical2} and \cref{sec:identical3}).

Budish~\cite{Budish11} has suggested the notion of MMS as a relaxation of proportionality:
\begin{definition}[Maximin‐Share (MMS)]
Given a set of goods $\mathcal{M}$ and a number of agents $n$, the \emph{maximin share (MMS)} of an agent with valuation $v$, is defined to be:
\[ \textnormal{MMS}(v, n, \mathcal{M}) = \max_{(A_1,\dots,A_n)\in\Pi(n, \mathcal{M})} \min_{l\in[n]} v_i(A_l). \]
An allocation $(A_1,\dots,A_n)$ is \emph{MMS‐fair} (or simply MMS) if $v_i(A_i)\ge\textnormal{MMS}(v_i, n, \mathcal{M})$ for all $i\in \mathcal{N}$. 
Moreover, for $\rho\in(0,1]$ we say that an allocation is 
\emph{$\rho$-MMS-fair} (or is a \emph{$\rho$-approximation} to the MMS)
 if $v_i(A_i)\ge\rho\cdot\textnormal{MMS}(v_i, n, \mathcal{M})$ for every 
 $i\in \mathcal{N}$. {When for $\rho>0$ an allocation is 
$\rho$-MMS-fair then we say that it provides a \emph{constant-factor approximation} to the MMS ex-post.}
\end{definition}
When the context of the number of agents $n$, the set of goods $\mathcal{M}$, and the vector of valuations $\vec{v}= (v_1,v_2, \ldots, v_n)$ is clear, we simply write
$\textnormal{MMS}_i = \textnormal{MMS}(v_i, n, \mathcal{M})$ for the maximin share of agent $i$ with valuation $v_i$ over $\mathcal{M}$.
Moreover, when the valuations are identical, we simply write $\textnormal{MMS}$.

EF1 was suggested in \cite{LiptonMMS04} (and 
formally defined by Budish~\cite{Budish11}) 
as a relaxation of envy-freeness:
\begin{definition}[EF1]
An allocation $(A_1,\dots,A_n)$ is envy‐free-up-to-one-good (EF1) if for every pair $i,l\in \mathcal{N}$, 
\emph{there exists} a good $g\in A_l$,
\[ v_i(A_i) \ge v_i(A_l\setminus\{g\}). \]
\end{definition}

EFX was suggested in \cite{caragiannis2019unreasonable} 
as a relaxation of envy-freeness and a strengthening of 
EF1:
\begin{definition}[EFX]
An allocation $(A_1,\dots,A_n)$ is envy‐free-up-to-any-good (EFX) if for every pair $i,l\in \mathcal{N}$ and \emph{every good} $g\in A_l$,
\[ v_i(A_i) \ge v_i(A_l\setminus\{g\}). \]
\end{definition}
Amanatidis et al.
\cite{AmanatidisBirmpasMarkakis2018} showed that
for identical valuations, EFX implies $2 / 3$-MMS for $n = 2,3$ and 
$4 / 7$-MMS for $ n \geq 4$.

\subsection{Allocations and Distributions}

Recall that $\Pi(n, \mathcal{M})$ denotes the set of all allocations of the goods $\mathcal{M}$ to the agents
$\mathcal{N}$.  
We define the \emph{value-vector of an allocation 
$A=(A_1,\dots,A_n)$ by valuations $(v_1,v_2,\ldots, v_n)$} 
 to be $\bigl(v_1(A_1),\dots,v_n(A_n)\bigr)$. When all valuations are identical and equal to $v$ we denote this value-vector by $v(A)$.
A distribution $D$ over allocations is a probability mass function $D:\Pi(n, \mathcal{M})\to[0,1]$ 
where $\sum_{A\in\Pi(n, \mathcal{M})} D(A)=1$. 
{We} 
denote the set of all distributions over $\Pi(n, \mathcal{M})$ by $\Delta(\Pi(n, \mathcal{M}))$.
The support of $D$ is defined as $\textnormal{Supp}(D) = \{A \in \Pi(n, \mathcal{M}) \mid D(A)>0\}$.
Denote the size of the support of $D$ by $k$. 
We use $\numberedAllocation{j}$ to denote the $j$-th allocation in the support, 
so $ \textnormal{Supp}(D) = 
\left\{ \numberedAllocation{j} \right\}_{j=1}^{k}$. We denote the probability that 
$D$ assigns to allocation 
$\numberedAllocation{j} \in \textnormal{Supp}(D)$ 
by $p_j$ (the probability is 0 for any other allocation). 

The \emph{expected value of agent $i$ 
under $D$} is $\mu_i(D)
= \mathbb{E}_{(A_1, \ldots, A_n) \sim D}[v_i(A_i)]=\sum_{A=(A_1,\dots,A_n)}D(A)\,v_i(A_i)$. 
The \emph{variance of agent $i$ under $D$} is 
$
\Var_{(A_1, \ldots, A_n) \sim D}[v_i(A_i)] = \sum_{A=(A_1,\dots,A_n)}D(A)\,(v_i(A_i) - \mu_i(D)
)^2$. 
{
The \emph{Standard Deviation of agent $i$ under $D$} is
$\sigma_i(D) = \sqrt{\Var_{(A_1, \ldots, A_n) \sim D}[v_i(A_i)]}$.}

For valuation vector $\vec{v}= (v_1,v_2, \ldots, v_n)$, a distribution $D$ is \emph{ex‐ante proportional} if $\mu_i(D)
\ge\text{PS}(v_i, n, \mathcal{M})$ for every $i\in \mathcal{N}$, 
and we denote the set of all ex-ante proportional distributions for $\vec{v}$ by $PROP_{\vec{v}}$.
For any ex-post property (as EFX), we say that  distribution $D$ satisfies this property if the property holds for any allocation in the support of the distribution $D$.

\subsection{Sum‐of‐variances Objective}
Consider a setting with $n$ agents with a valuation vector $\vec{v}= (v_1,v_2, \ldots, v_n)$ over a set of items $\mathcal{M}$.
For a distribution $D$ over the set $\Pi(n, \mathcal{M})$ of allocations of $\mathcal{M}$ to $n$ agents, 
define the sum-of-variances (SoV) objective $\Phi_{\vec{v}}(\cdot)$
to be: 
\[ \Phi_{\vec{v}}(D) = \sum_{i=1}^n \Var_{(A_1, \ldots, A_n) \sim D}[v_i(A_i)]. \]
We consider minimizing this objective under the constraint that $D\in PROP_{\vec{v}}$, that is, $D$ is ex‐ante proportional for $\vec{v}$.

A distribution $D$ is called an 
\emph{ex-ante proportional sum-of-variances-minimizing 
distribution (for $\vec{v}$)} if the distribution $D$ is both ex-ante proportional for $\vec{v}$, 
and minimizes the sum-of-variances objective $\Phi_{\vec{v}}$ over all ex-ante proportional distributions for $\vec{v}$. 
We use $\mathcal{D}^{\Phi}_{\vec{v}}$ to denote the set of all ex-ante proportional sum-of-variances 
minimizing distributions.
  Formally, for valuation vector $\vec{v}= (v_1,v_2, \ldots, v_n)$, 
  recall that $PROP_{\vec{v}}$ is the set of ex-ante proportional 
  valuations. We define the set 
  $\mathcal{D}^{\Phi}_{\vec{v}}$ of ex-ante proportional 
  sum-of-variances-minimizing distributions
  as follows: 
  \begin{align*}
    \mathcal{D}^{\Phi}_{\vec{v}} = \left\{ D\in \Delta(\Pi(n, \mathcal{M}))\mid D\in PROP_{\vec{v}}\text{, and } \Phi_{\vec{v}}(D) \leq \Phi_{\vec{v}}(D') \text{ for every } D'\in PROP_{\vec{v}} \right\}.
  \end{align*}
Observe that the set $\mathcal{D}^{\Phi}_{\vec{v}}$ is not empty. 
Indeed, the set $PROP_{\vec{v}}$ of ex-ante proportional distributions is compact and 
non-empty. 
Moreover, 
there exists a minimizer $D^\ast \in \mathcal{D}^{\Phi}_{\vec{v}}$ of $\Phi_{\vec{v}}$ over the set 
$PROP_{\vec{v}}$,
as the sum-of-variances objective $\Phi_{\vec{v}}$ is continuous on the non-empty compact set $PROP_{\vec{v}}$.
Finally, when ${\vec{v}}$  is clear from the context we use $\Phi$ to denote  $\Phi_{\vec{v}}$, and use $\mathcal{D}$ to denote $\mathcal{D}^{\Phi}_{\vec{v}}$.

We note that one can also consider the sum-of-variances-of-ratios (SoVoR) objective, a variant of the sum-of-variances (SoV) objective. 
In this variant, instead of considering the objective 
for the (absolute) value of the allocation to the agents 
(the value $v_i(A_i)$ for agent $i$), we consider the 
objective for  the fraction of the value obtained 
(out of the value of the entire set of items) by the 
agents (the fraction $v_i(A_i)/v_i(\mathcal{M})$ for 
agent $i$).  

Thus, the \emph{sum-of-variances-of-ratios (SoVoR)} 
objective would be
\begin{align*}
\tilde{\Phi}_{\vec{v}}(D) = \sum_{i=1}^n \Var_{(A_1, \ldots, A_n) \sim D}\left[\frac{v_i(A_i)}{v_i(\mathcal{M})}\right], 
\end{align*}
where $\Var_{(A_1, \ldots, A_n) \sim D}\left[\frac{v_i(A_i)}{v_i(\mathcal{M})}\right]
= \sum_{A=(A_1,\dots,A_n)}D(A)\cdot \left(\frac{v_i(A_i)}{v_i(\mathcal{M})} - \tilde{\mu}_i(D)\right)^2$ is the variances of the ratios, and $\tilde{\mu}_i(D) = \mathbb{E}_{(B_1, \ldots, B_n) \sim D}\left[\frac{v_i(B_i)}{v_i(\mathcal{M})}\right]$ is the expected value of the ratios.

When valuations are normalized ($v_l(\mathcal{M})=v_i(\mathcal{M})$ for every $i,l\in \mathcal{N}$) the problem of minimizing the 
sum-of-variances-of-ratios (SoVoR) is equivalent to the 
problem of minimizing the sum-of-variances (SoV), 
and the set of solutions is identical. 
As our results are all for normalized valuations, 
they apply for both the SoV and the SoVoR objectives.

Finally, throughout the paper we 
use $\textnormal{Ber}(p)$ to denote a Bernoulli random 
variable with parameter $p$ (the random variable 
$\textnormal{Ber}(p)$ has value of 1 with probability 
$p$, and 0 otherwise). We 
occasionally compute 
the variance of a random variable $X$ 
which receives the value 
$b$ with probability $p$, 
and the value $a$ otherwise. 
Thus, the random variable 
distributes according to $X \sim (b - a) 
\cdot \textnormal{Ber}(p) + a$. 
Recall that for such a random variable, $\Var[X] = p (1-p) \cdot {(b - a)^2}$. 

\section{Two Agents with Identical Additive Valuations}\label{sec:identical2}

In this section, we focus on the simple case of two agents with identical additive valuations. 
We observe  that for this simple case, every ex-ante proportional 
sum-of-variances-minimizing distribution $D\in \mathcal{D}$ is MMS-fair ex-post.

Throughout this section, let $v$ be the additive valuation that is common to both agents.
As agents are identical, they have the same maximin share, and we denote it by $\MMS$. Similarly, they have the same proportional share, denoted by  $\text{PS} = v(\mathcal{M})/2$.

{Observe that for $n$ players with an identical additive valuation,
every ex-ante proportional distribution $D$ provides each agent exactly the proportional share in expectation,
i.e. $\mu_i(D) 
= \text{PS} = v(\mathcal{M}) / n$ for every $1\leq i\leq n$.} 
We first show that any allocation in the support of an ex-ante proportional sum‐of‐variances-minimizing distribution is ``difference minimizing'':

\begin{definition}[Difference-Minimizing Allocation]\label{def:diff-minimizing}
An allocation $S=(S_1,S_2)$ is \emph{difference minimizing} 
for a valuation $v$ over $\mathcal{M}$, if it minimizes 
$|v(S_1)-v(S_2)|$ over all allocations of $\mathcal{M}$ 
into two bundles.
\end{definition}

\begin{restatable}[Support on Difference‐Minimizing Allocations]{lemma}{suppdiffmin}
\label{lem:diff-minimizing}
Consider a setting with two agents with the same additive valuation $v$ over a set of goods $\mathcal{M}$.  
Let $D\in \mathcal{D}$ be any  ex-ante proportional sum‐of‐variances-minimizing distribution over allocations of $\mathcal{M}$
to two agents.
Then every allocation in the support of $D$ is difference minimizing.
\end{restatable}

\textit{Proof deferred to \cref{app:proofs}.}

It is well known that when partitioning to two bundles, 
a difference‐minimizing allocation for $v$ implies that 
the value of every bundle is at least the MMS for $v$ 
\cite{Korf10}. 

\begin{restatable}[Difference‐Minimizing Allocation $\implies$ MMS‐Fair]{lemma}{diffminimizingmmsfair}
\label{lem:diff-to-mms}
Any difference‐minimizing allocation $S=(S_1,S_2)$ satisfies $v(S_1)\ge\MMS$ and $v(S_2)\ge\MMS$. 
\end{restatable}

\textit{For completeness we  present the proof of this known result in \cref{app:proofs}.}

An immediate corollary of \cref{lem:diff-minimizing} 
and \cref{lem:diff-to-mms} is that 
any ex-ante proportional sum‐of‐variances-minimizing 
distribution is supported only on difference‐minimizing 
allocations, and thus is MMS‐fair ex‐post. 
Moreover, as 
MMS-fairness implies EFX, any allocation in the support is EFX ex-post.

\begin{corollary}[Ex‐Post MMS‐Fairness of Sum‐Of‐Variances Minimizer]\label{cor:identical2-mms-exact}
Consider a setting with two agents with the same additive valuation $v$ over a set of goods $\mathcal{M}$.
Every allocation in the support of an ex-ante proportional sum‐of‐variances-minimizing distribution $D\in \mathcal{D}$ is MMS‐fair ex‐post (and thus also EFX ex-post).
\end{corollary}

We observe that if given an additive  valuation $v$ one can efficiently compute an ex-ante proportional sum‐of‐variances-minimizing distribution $D\in \mathcal{D}$ for two agents 
with valuation $v$, then one can efficiently compute an MMS 
allocation for an instance with two agents, both with 
valuation $v$ (by looking at any allocation in the support 
of $D$).
An immediate corollary from the hardness of MMS computation 
is the following:

\begin{corollary}[NP‐Hardness]\label{prop:np-hard}
  Consider the computation problem of finding, 
  for an additive valuation $v$, a distribution $D$ 
  which minimizes $\Phi_{(v,v)}(D)$ over all ex-ante 
  proportional distributions 
  (computing it for a  
  setting with two agents with the same valuation $v$). 
  This problem is NP-hard.
\end{corollary}

\section{Multiple Agents with Identical Additive Valuations}\label{sec:identical3}
In this section we consider $n\geq 3$ agents with identical additive valuations. 
For the case of two agents, we have observed that any ex-ante proportional sum‐of‐variances-minimizing distribution is MMS-fair ex-post (\cref{cor:identical2-mms-exact}). 
In this section we show that such a guarantee does not hold for $n\geq 3$ agents.
That is, we show that in some settings with $n\geq 3$ agents with \emph{identical} additive valuations, 
any ex-ante proportional sum‐of‐variances-minimizing distribution is not MMS-fair ex-post. 
This is particularly disappointing given that an ex-ante proportional distribution that is MMS-fair ex-post does exist for any setting with $n$ agents with \emph{identical} additive valuations (an MMS allocation exists when valuations are identical, 
and randomly picking any one of the $n$ cyclic permutations of such allocation is an ex-ante proportional distribution).
Nevertheless, we show that any ex-ante proportional sum‐of‐variances-minimizing distribution is EFX ex-post, which implies MMS approximation ex-post ($2/3$ for three agents, and $4/7$ for more agents \cite{AmanatidisBirmpasMarkakis2018}).\footnote{Note that as all valuations are identical and as we allocate all items, any allocation is also Pareto optimal.} 
Finally, we show that the MMS approximation obtained for three agents is no better than $275 / 304 \approx 90.4\%$. 

  We begin by presenting a structural characterization 
  of the allocations that can appear in the support of 
  any ex-ante proportional sum-of-variances-minimizing distribution. 
  We consider mapping every distribution to a point in $\mathbb{R}_{\geq 0}^n$ that corresponds to the expected values of the $n$ agents. 
  For example, when all agents have valuation $v$, the \emph{ex-ante proportional-share vector} is mapped to the $n$ dimensional vector with all entries being the proportional share $v(\mathcal{M})/n$.
  We show that the sum-of-variances of an 
  ex-ante proportional distribution can be expressed as 
  the weighted average of the 
  squared Euclidean distances (L2 norm) of 
  its supporting allocations from the ex-ante proportional-share vector  $(\text{PS})^n$.
  {Moreover, 
  for any allocation with distance $d$ from $(\text{PS})^n$ there exists an ex-ante proportional distribution
  which is supported only on allocations that have distance $d$ from $(\text{PS})^n$   (by randomly picking one
  of the $n$ cyclic permutations of this allocation)}.
  Thus,
  we have that every ex-ante proportional
  sum-of-variances-minimizing distribution is supported only on allocations that minimize   the squared distance.\footnote{ 
  {
    As the square function is monotone,
    minimizing the square of the distance 
    to $(\textnormal{PS})^n$ 
    is equivalent to minimizing 
    the distance itself.
    Thus, any ex-ante proportional sum-of-variances-minimizing  
    distribution is 
    supported only on allocations that minimize 
    the Euclidean distance.
    }
  }
  Building on this characterization, we prove 
  that every such distance-minimizing  
  allocation must be 
  EFX, and therefore every
  ex-ante proportional sum-of-variances-minimizing 
  distribution 
  is EFX ex-post. This, in turn, immediately yields a positive 
  constant-factor approximation to the MMS 
  \cite{AmanatidisBirmpasMarkakis2018}.
  On the negative side, we show that sum-of-variances 
  minimization does 
  not necessarily guarantee MMS-fairness once the 
  number of agents exceeds two, by providing a 
  concrete counterexample with three agents. This 
  example also establishes an explicit upper bound 
  on the MMS approximation achievable ex-post.

\subsection{Structural Characterization of the Support}

Recall that for $n$ players with an identical additive valuation,
$\text{PS} = v(\mathcal{M})/n$ denotes the proportional share of every agent. Also recall that 
every ex-ante proportional distribution $D$ provides each agent with exactly the proportional share in expectation,
i.e. $\mathbb{E}_{(A_1,\ldots, A_n)\sim D}[v(A_i)] = \text{PS} = v(\mathcal{M}) / n$ for every $1\leq i\leq n$.

We next show that $\Phi(D)$ can be expressed as weighted 
sum of squared distances (with L2 norm) of the supporting allocations from the proportional share 
vector $(\text{PS})^n$.

\begin{lemma}[SoV Characterization]\label{lem:variance-cyclic}
Consider a setting with $n$ agents that have the same additive valuation $v$ over $\mathcal{M}$. 
If a distribution $D$ is ex-ante proportional for $\vec{v}=(v,v,\ldots,v)$ then: 

\begin{align*}
  \Phi(D) = \sum_{j=1}^{k} p_j \cdot || v(\numberedAllocation{j}) - (\textnormal{PS})^n ||^2, 
\end{align*} 

where $\textnormal{Supp}(D) = \{ \numberedAllocation{j} \}_{j=1}^{k}$,
and $p_j>0$ is the probability that $D$ assigns to $\numberedAllocation{j}$.
 
\end{lemma}
\begin{proof}
  As $D$ is ex-ante proportional, we have that $\mathbb{E}_{(A_1,\ldots, A_n)\sim D}[v(A_i)] = \textnormal{PS}$, 
  for every $i\in \mathcal{N}$.
  Thus, for all $i \in \mathcal{N}$ we have that:
  \begin{align*}
    \textnormal{Var}_{(A_1,\ldots, A_n)\sim D}[v(A_i)] 
    &= \sum_{j=1}^{k} p_j \cdot \left( v\left(\numberedBundle{j}{i}\right) - \textnormal{PS} \right)^2, \\
  \end{align*}
  where $\numberedBundle{j}{i}$ is the bundle allocated to agent $i$ in allocation $\numberedAllocation{j}$. 
  Thus,
  \begin{align*}
    \Phi(D) = \sum_{i=1}^{n} \textnormal{Var}_{(A_1,\ldots, A_n)\sim D}[v(A_i)] 
    &= \sum_{i=1}^{n} \sum_{j=1}^{k} p_j \cdot \left( v\left(\numberedBundle{j}{i}\right) - \textnormal{PS} \right)^2, \\
    &= \sum_{j=1}^{k} p_j \cdot  \sum_{i=1}^{n} \left( v\left(\numberedBundle{j}{i}\right) - \textnormal{PS} \right)^2, \\
    &= \sum_{j=1}^{k} p_j \cdot || v\left(\numberedAllocation{j}\right) - (\textnormal{PS})^n ||^2,
  \end{align*}
  and the claim follows.
\end{proof}

\begin{corollary}[Distance Minimizing Support]\label{cor:identical3-sov-minimizing}
Consider a setting with $n$ agents that have the same additive valuation $v$ over $\mathcal{M}$. 
Let $D\in \mathcal{D}$ be an ex-ante proportional sum‐of‐variances-minimizing distribution (minimizes $\Phi(D)$ subject to ex-ante proportionality). 
  Then, for every allocation $P \in \textnormal{Supp}(D)$ we have that:
  \begin{align*}
    || v(P) - (\textnormal{PS})^n ||^2 
    = \min_{P' \in \Pi(n, \mathcal{M})} || v(P') - (\textnormal{PS})^n ||^2.
  \end{align*} 
  Moreover,
  \begin{align*}
    \Phi(D) = \min_{P' \in \Pi(n, \mathcal{M})} || v(P') - (\textnormal{PS})^n ||^2.
  \end{align*}
\end{corollary}

\begin{proof}
  Suppose towards contradiction that there exists some
  allocation $P \in \textnormal{Supp}(D)$
  such that $P$ does not minimize $|| v(P) - (\textnormal{PS})^n ||^2$.
  Let $\numberedAllocation{min} = \textnormal{argmin}_{P' \in \Pi(n, \mathcal{M})} || v(P) - (\textnormal{PS})^n ||^2$. 
  Using $\numberedAllocation{min}$ we define 
  the distribution $D'$ which receives the 
  $n$ cyclic 
  shifts of $\numberedAllocation{min}$ with probability 
  $1 / n$ each.\footnote{Formally, 
  for allocation $A = (A_1, A_2, \ldots, A_{n})$, the cyclic shift  by $k$ positions is defined as   $A^{(k)} = (A_{1+(k \mod n)}, A_{1+((k + 1) \mod n)}, \ldots, A_{1+ ((k + n - 1) \mod n)})$. The $n$ cyclic shifts of $A$ are the $n$ allocations $A^{(1)}, \ldots, A^{(n)}$.}
  Note that $D'$ is clearly ex-ante proportional.  
  By \cref{lem:variance-cyclic}, we have that:
  \begin{align*}
    \Phi(D) =
     \sum_{j=1}^{k} p_j \cdot || v(\numberedAllocation{j}) - (\textnormal{PS})^n ||^2 
    > || v(\numberedAllocation{min}) - (\textnormal{PS})^n ||^2 = \Phi(D'),
  \end{align*}
  where $ \textnormal{Supp}(D) = \left\{ \numberedAllocation{j} \right\}_{j=1}^{k}$ and $p_j$ is the probability that $D$ 
  assigns to allocation $\numberedAllocation{j}$.  
  A contradiction arises due to the fact that $D$ is an ex-ante proportional sum‐of‐variances-minimizing distribution
  and the first part of the claim follows.
  The second part follows immediately from \cref{lem:variance-cyclic} (as $\sum_{j=1}^k p_j=1$),
  combined with the first part of the claim.
\end{proof}

{
Let us illustrate this for the case of $n=3$ agents with an identical valuation $v$, 
with value of $\mathcal{M}$ being normalized to 1 ($v(\mathcal{M}) = 1$).
Observe that each value-vector $v(A)$ 
of an allocation $A = (A_1, A_2, A_3)$
by valuations $\vec{v} = (v,v,v)$ 
is a vector $v(A) = (v(A_1), v(A_2), v(A_3))$
such that $v(A_1) + v(A_2) + v(A_3) = v(\mathcal{M}) = 1$.  
This allows us to visualize the value-vectors as 
points on the 2-dimensional simplex $\Delta^2$ in $\mathbb{R}_{\geq 0}^3$
as done in \cref{fig:3Dtriangle}.
As the sum-of-variances of an ex-ante proportional distribution
only depends on the distances of the value-vectors
of its supporting allocations
(and the probabilities assigned to them) from the 
proportional share vector $(\text{PS})^3= (1/3,1/3,1/3)$,
we visualize 
the value-vectors of the allocations according to their Euclidean
distance from 
$(\text{PS})^3$.
Moreover, the support of an ex-ante proportional sum-of-variances-minimizing distribution
will necessarily only contain allocations which minimize this distance, allocations
which correspond to the innermost circle in\footnote{ 
  We note that the figure is only illustrative and suits 
  the example in \cref{prop:identical3-not-exact}.
  In particular, there can be MMS-fair allocations with 
  different distances from the center. 
} \cref{fig:3Dtriangle}.}

We can prove a similar claim when the support is restricted to only include allocations that are MMS-fair ex-post. 
This will allow us to show that 
every ex-ante proportional
sum-of-variances-minimizing distribution is not MMS-fair ex-post.

\begin{corollary}\label{cor:identical3-sov-minimizing-mms}
    Consider a setting with $n$ agents that have the same additive valuation $v$ over $\mathcal{M}$. 
    Let $D\in \mathcal{D}$ be an ex-ante proportional distribution which minimizes $\Phi(D)$ 
    over distributions in $\mathcal{D}$ {\bf that are ex-post MMS-fair}.
  Then, for every allocation $P \in \textnormal{Supp}(D)$ we have that:
  \begin{align*}
    || v(P) - (\textnormal{PS})^n ||^2 
    = \min_{P' \textnormal{ is MMS}} || v(P') - (\textnormal{PS})^n ||^2.
  \end{align*} 
  Moreover,
  \begin{align*}
    \Phi(D) = \min_{P' \textnormal{ is MMS}} || v(P') - (\textnormal{PS})^n ||^2.
  \end{align*}
\end{corollary}

The proof of this claim is similar to that of \cref{cor:identical3-sov-minimizing},
and we omit it here.

\subsection{Counterexample: SoV Minimization does not imply MMS Ex‐Post}
While we have seen  
that for two agents with identical additive valuations 
every ex-ante proportional 
sum-of-variances-minimizing distribution is 
MMS‐fair ex‐post (\cref{cor:identical2-mms-exact}), 
we next present an example showing this is not 
the case once moving to three agents. To show this, we present an explicit 
example with three agents and identical additive 
valuations for which every ex-ante proportional 
sum-of-variances-minimizing distribution is not MMS‐fair 
ex‐post.

{
  Essentially, by \cref{cor:identical3-sov-minimizing-mms},
 ex-ante proportional sum-of-variances-minimizing
  distributions are
  MMS-fair ex-post if and only 
  if there exists an MMS allocation
  with a value-vector that minimizes 
  the distance to the proportional share vector. However, 
  since the MMS is defined in terms of maximizing the 
  least-valued bundle, it is natural that other 
  allocations—ones that need not be MMS—can achieve 
  the smallest distance from the proportional share vector,
  as illustrated in \cref{fig:3Dtriangle}.
  This suggests that sum-of-variances-minimization does 
  not, in general, yield MMS-fairness ex-post. 
  This distinction also clarifies why, in the special 
  case of two agents with identical additive valuations, 
  the ex-ante proportional 
  sum-of-variances-minimizing distributions are 
  MMS-fair ex-post: in that setting, the allocations
  with a value-vector
  minimizing the distance to the proportional share 
  vector coincide exactly with those that maximize 
  the minimum bundle value.
}

\begin{proposition}[Failure of Exact MMS Ex‐Post]
  \label{prop:identical3-not-exact}
There exists an instance of three
agents with identical additive valuations over six goods,
where every ex-ante proportional sum‐of‐variances-minimizing distribution 
$D\in \mathcal{D}$
(minimizes $\Phi(D)$ subject to ex-ante proportionality)  is not MMS‐fair ex‐post.
\end{proposition}

\begin{proof}
 The proof follows from analyzing the following example
 {(which was found using a computer simulation):}
 \begin{example}\label{example:3identical-lower-bound}
 Consider the case of three agents
  and let $v:\mathcal{M} \to \mathbb{R}_{> 0}$ be the additive  
  valuation function over the set $\mathcal{M}=\{a,b,c,d,e,f\}$ of $6$ items, defined as follows:
  \begin{align*}
    v(a) = 825, \quad 
    v(b) = 552, \\
    v(c) = 528, \quad
    v(d) = 384, \\
    v(e) = 360, \quad
    v(f) = 351. \\
  \end{align*}
  \end{example}

  For three agents and the valuation $v$, there is only one MMS-fair allocation of $\mathcal{M}$
  which is
  $ \numberedAllocation{\textnormal{MMS}} = ( \left\{ a,f \right\}, \left\{ b,e \right\}, \left\{ c,d \right\})$
  where $v(\numberedAllocation{\textnormal{MMS}}) = (1176, 912, 912)$ (so the MMS is 912),
  and $|| v(\numberedAllocation{\textnormal{MMS}}) - ( \textnormal{PS}, \text{PS}, \text{PS}) ||^2 = 46664$.
  Iterating over all possible allocations shows that 
  the only distance minimizing allocation (up to permutations) is 
  $ \numberedAllocation{min} = ( \left\{ d,e,f \right\}, \left\{ b,c \right\}, \left\{ a \right\})$
  where $v(\numberedAllocation{min}) = (1095, 1080, 825)$
  and $|| v(\numberedAllocation{min}) - ( \text{PS}, \text{PS}, \text{PS}) ||^2 = 46050$.
  Thus, every 
  ex-ante proportional sum‐of‐variances-minimizing distribution $D\in \mathcal{D}$ 
  only has permutations of $\numberedAllocation{min}$ in its support.
  Finally, notice that
  $\min_{i\in [3]} v(\numberedBundle{min}{i}) = 825 < 912 = \textnormal{MMS}$
  and thus every ex-ante proportional sum‐of‐variances-minimizing distribution 
  $D\in \mathcal{D}$ 
  is not MMS-fair ex‐post.
\end{proof}

\begin{figure}[htbp]
  \centering
  \includegraphics[width=0.8\textwidth]{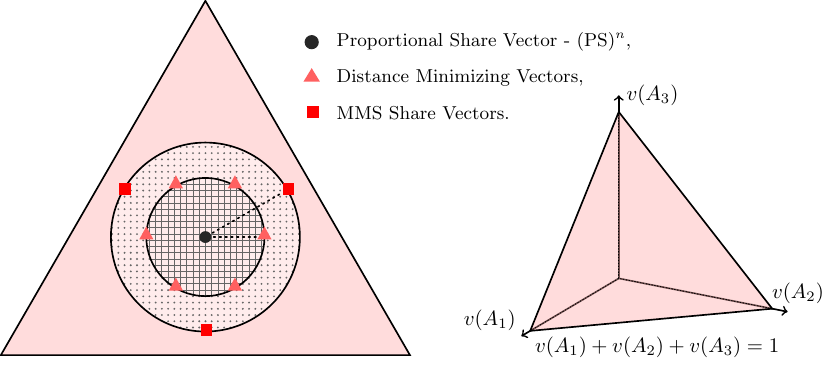}
  \caption{Visualization of \cref{example:3identical-lower-bound} (used to prove  \cref{prop:identical3-not-exact}), showing allocations value vectors and their distances from the proportional share vector. Not to scale. Only relevant allocations are shown. }
  \label{fig:3Dtriangle}
\end{figure}

\cref{fig:3Dtriangle} illustrates the value-vectors of different possible allocations 
in \cref{example:3identical-lower-bound} (used to prove \cref{prop:identical3-not-exact}),
along with their distances from the proportional
share vector\footnote{ 
  Because two bundles in the MMS-fair allocation have the same value,
  there are only three possible permutations of the MMS-fair allocation. 
  For the distance-minimizing allocation,
  there are six possible permutations.
} $(\textnormal{PS})^3$.
Crucially, we selectively chose which 
allocations to display in the figure,
{we only display the distance-minimizing allocations
and the MMS-fair allocations, 
and, for clarity, omit drawing allocations that are irrelevant for the argument.}
As in \cref{example:3identical-lower-bound} 
there exists a unique MMS-fair allocation
(up to permutations), we can divide the triangle $\Delta^2$ 
into three regions. 
(1) The first is the inner region (disk), 
which is marked with a grid. 
This region contains no value-vectors, as 
a value-vector in this region
would imply the existence of an allocation 
with a value-vector 
which is closer to $( \textnormal{PS})^3$
than that of 
the distance-minimizing allocation.
(2) The second region is the ``middle'' region (ring), 
which is dotted.  
This region contains all value-vectors 
which are closer to $( \textnormal{PS})^3$ than
the MMS-fair allocation,\footnote{ 
  As there is a unique MMS-fair allocation
  (up to permutations) all MMS-fair allocations
  have the same distance from $( \textnormal{PS})^3$.
}
the existence of allocations in this region
implies that every ex-ante proportional
sum-of-variances-minimizing distribution
is not MMS-fair ex-post. 
Finally, there exists the outer region, which is 
solid. 
This region contains the value-vector of every other allocation.

\subsection{EFX ex-post and MMS Approximation}

{
We have seen that for three agents, 
an ex-ante proportional distribution which minimizes 
the sum-of-variances might not be MMS-fair ex-post. 
Yet, we next show that for any $n\geq 2$, any 
ex-ante proportional sum-of-variances-minimizing 
distribution is EFX ex-post. 
By prior results \cite{AmanatidisBirmpasMarkakis2018} this 
implies a constant MMS approximation guarantees 
(fraction of $2/3$ for three agents, and of $4/7$ for more agents).
Additionally, we use the previously shown example to 
provide an upper bound 
of $ 275 / 304 \approx 0.904$ 
on the best possible MMS approximation ex-post.}

\begin{theorem}[EFX Ex‐Post]\label{thm:identical3-efx}
Consider a setting with $n\geq 2$ agents that have the same additive valuation $v$ over $\mathcal{M}$, 
and all goods have positive value.\footnote{One can also allow zero-valued items but in that case, 
adjust the definition of EFX to only require envy to 
be eliminated when positive-valued items are removed 
(as in the original definition of \cite{caragiannis2019unreasonable}). 
Alternatively, the claim without this restriction 
(to positive-valued items) will hold if assuming that 
the ex-ante proportional distribution sum-of-variances-minimizing distribution always puts all zero-valued items in the least valuable set.} 
Let $D\in \mathcal{D}$ be an ex-ante proportional sum‐of‐variances-minimizing distribution. 
Then, \(D\) is EFX ex‐post.
\end{theorem}

\begin{proof}
  Assume w.l.o.g that $v(\mathcal{M}) = 1$.
  Let $B \in \textnormal{Supp}(D)$
  be some allocation in the support of $D$
  and assume for the sake of contradiction that $B$ is not EFX.
  Due to \cref{cor:identical3-sov-minimizing},
  $B$ must minimize the distance to the 
  $( \textnormal{PS})^n$
  over all allocations in $\Pi(n, \mathcal{M})$. 
  Denote $B = (B_1, \ldots, B_n)$
  and assume w.l.o.g that 
  $v(B_1) \geq v(B_2) \geq ... \geq v(B_n)$.
  As the allocation is not EFX, there exists an agent $i \in \mathcal{N}$
  and good $g \in B_i$ such that 
  $v(B_i \setminus \left\{ g \right\}) > v(B_n)$.
  We create a new allocation 
  $B' = (B_1, \ldots, B_{i-1}, B'_i, B_{i+1},\ldots, B_{n-1}, B'_n)$,
  where $B'_i = B_i \setminus \left\{ g \right\}$
  and $B'_n = B_n \cup \left\{ g \right\}$ (while for $l\neq i,n$ we have $B'_l=B_l$). 
  We prove that $B'$ has a shorter distance to $(\textnormal{PS})^n$, in contradiction 
  to \cref{cor:identical3-sov-minimizing}. 
  Let $V = 1-v(B_i)-v(B_n)$  
  and let 
  $C = \sum_{l\in [n]\setminus \{i,n\}} (v(B_l) - \textnormal{PS})^2$. 
  We define:
  \begin{align*}
    f(x) 
    &= || (v(B_1), \ldots, v(B_{i-1}), x, v(B_{i+1}), \ldots, v(B_{n-1}), 1 - V - x) - ( \textnormal{PS})^n ||^2 \\
    &= (x - \textnormal{PS})^2 + (1 - V - x - \textnormal{PS})^2 + C
  \end{align*}
  
  Notice that $v(B_i) + v(B_n) = v(B'_i) + v(B'_n) = 1 - V$. 
  It holds that $v(B_n) <v(B_n)+ v({g})= v(B'_n)$ and $v(B'_i) =v(B_i)- v({g})< v(B_i)$, 
  by the assumptions that $v(\{g\})>0$ (all items have positive value). 
  Additionally, we assume by contradiction that $v(B'_i)= v(B_i \setminus \left\{ g \right\}) > v(B_n)$, 
  thus $v(B_i) > v(B_n) + v(\{g\}) = v(B'_n)$. We conclude that 
  $ v(B_n) < \min\left(v(B'_n), v(B'_i)\right) $ and $\max\left(v(B'_n ), v(B'_i )\right)< v(B_i)$, so:
  \begin{align}
    \label{eq:inequality}
    v(B_n) 
    < \min\left(v(B'_n), v(B'_i)\right) 
    \leq \max\left(v(B'_n ), v(B'_i )\right)
    < v(B_i). 
  \end{align}
  
  Notice that $f(v(B_i)) = f(v(B_n))$ and $f(v(B'_i)) = f(v(B'_n))$.
  As the function $f(x)$ is strictly convex 
  and due to \cref{eq:inequality},
  we have that 
  $f(v(B'_i)) =\frac{f(v(B'_i)) + f(v(B'_n))}{2}  < \frac{f(v(B_i)) + f(v(B_n))}{2} =f(v(B_i))$. 
  Thus,
  \begin{align*}
    || v(B') - (\textnormal{PS})^n ||^2 
    = f(v(B'_i)) 
    < f(v(B_i)) 
    = || v(B) - (\textnormal{PS})^n ||^2,
  \end{align*}
  which is a contradiction to \cref{cor:identical3-sov-minimizing}. 
  The claim follows.
\end{proof}

It was shown  \cite{AmanatidisBirmpasMarkakis2018}  that an EFX‐fair ex‐post 
allocation is ex‐post \(\tfrac{2}{3}\)\nobreakdash-MMS‐fair for three agents, 
and is ex‐post \(\tfrac{4}{7}\)\nobreakdash-MMS‐fair for $n\geq 4$. We thus 
derive immediately that while  ex-ante proportional sum‐of‐variances-minimizing distribution might not be ex-post MMS-fair (\cref{prop:identical3-not-exact}), it does provide the corresponding  MMS approximation guarantees 
when agents have identical additive valuations.
\begin{corollary}[MMS Ex‐Post Guarantees]\label{cor:identical3-mms}
Consider a setting with $n$ agents that have the same additive valuation $v$ over $\mathcal{M}$. 
Let $D\in \mathcal{D}$ be an  ex-ante proportional sum‐of‐variances-minimizing distribution. 
As $D$ is EFX‐fair ex‐post, the following maximin‐share (MMS) approximation guarantees hold (by \cite{AmanatidisBirmpasMarkakis2018}):
\begin{enumerate}
  \item For three  agents, \(D\) is ex‐post \(\tfrac{2}{3}\)\nobreakdash-MMS‐fair.
  \item For \(n \ge 4\)  agents, \(D\) is ex‐post \(\tfrac{4}{7}\)\nobreakdash-MMS‐fair.
\end{enumerate}
\end{corollary}

The following proposition provides an upper bound 
of $ 275 / 304 \approx 0.904$ 
to the approximation ratio 
which an ex-ante proportional sum‐of‐variances-minimizing distribution can provide to the MMS ex-post.

\begin{proposition}[MMS in-approximation]\label{prop:identical3-upper-bound}
There exists an instance with three agents with identical additive valuations over 
  six goods such that every ex-ante proportional sum-of-variances-minimizing distribution $D \in \mathcal{D}$ 
  gives some agent no more than $275 / 304 \approx 90.4\%$ of her MMS ex-post.
\end{proposition}

\begin{proof}
Consider the instance in \cref{example:3identical-lower-bound} (presented in the proof of \cref{prop:identical3-not-exact}). 
It was shown that $\textnormal{MMS} = 912$. 
Additionally, it was shown that every allocation in the support of an ex-ante proportional sum‐of‐variances-minimizing distribution $D \in \mathcal{D}$ is a permutation of 
$\numberedAllocation{min} = ( \left\{ d,e,f \right\}, \left\{ b,c \right\}, \left\{ a \right\})$
  where $v(\numberedAllocation{min}) = (1095, 1080, 825)$.
  As $\min_{i\in [3]} v(\numberedBundle{min}{i}) = 825$,  
  we conclude that any ex-ante proportional sum‐of‐variances-minimizing distribution $D$
  provides an approximation of 
  at most\footnote{{
    We remark that this counterexample is unlikely to be the example with the worst MMS approximation,  
    and additional work could establish better bounds.
  }} 
  $825/912 = 275 / 304$ to the $\MMS$ ex-post.
\end{proof}

\section{Main Result: Sum-of-Variances Minimization has Poor Fairness}
\label{sec:nonidentical2}

In this section we present our main negative result. We show that when valuations are not identical, there are settings in which any  ex-ante proportional sum‐of‐variances-minimizing distribution is not EF1 ex-post (thus also not EFX ex-post), and does not provide any  constant-factor   $\MMS$ approximation ex-post.
This absolute failure to provide any fairness ex-post guarantee occurs even in the most basic setting with only two agents with normalized additive valuations over only two goods. As the valuations are normalized, this negative result  holds with respect to the sum-of-variances-of-ratios (SoVoR) objective $\tilde{\Phi}(D)$ as well. 

\begin{theorem}[Failure for Non‐Identical Agents]\label{prop:nonidentical2-impossibility}
There exists an instance with two agents with additive valuations over two goods,
where both agents have a non-zero valuation for each good
and the valuations are normalized,
such that the following holds: 
the support of every  ex-ante proportional sum‐of‐variances-minimizing distribution 
includes an allocation which provides one of the agents an empty bundle. 
Thus,  any ex-ante proportional sum‐of‐variances-minimizing distribution 
$D \in \mathcal{D}$  
is not EF1 ex-post and provides no constant-factor approximation to the MMS ex-post. 
\end{theorem}
\begin{proof}
We prove that the claim holds for the following example:   
\begin{example}\label{example:lower-bound}
Consider a setting with two agents with the following additive valuations over the set $\mathcal{M} =\{a,b\}$ of two goods:
  \begin{align*}
    &v_1(a)=4, v_1(b)=96; \\
    &v_2(a)=38, v_2(b)=62.
  \end{align*}    
  In this setting there are four possible allocations of $\mathcal{M}$ into two bundles, 
  we enumerate them with their corresponding value-vectors:
  \begin{align*}
    \numberedAllocation{1} &=(\{a\},\{b\})\to(96,38), &  \numberedAllocation{2} &=(\{b\},\{a\})\to(4,62), \\
    \numberedAllocation{3} &=(\{a,b\},\emptyset)\to(100,0), &  \numberedAllocation{4} &=(\emptyset,\{a,b\}) \to(0,100). 
  \end{align*}

\end{example}
To prove the claim we analyze the ex-ante proportional sum‐of‐variances-minimizing distributions for this example.
  Note that the valuations are normalized  ($v_1(\mathcal{M})=v_2(\mathcal{M})=100$) 
  and that the MMS of every agent is positive 
  ($\textnormal{MMS}_1=4>0,\textnormal{MMS}_2=38 > 0$). 

  Consider the distribution\footnote{ 
    In fact, in \cref{prop:nonidentical2-sov-minimizing} we show that $D_{min}$ is the unique ex-ante proportional sum‐of‐variances-minimizing distribution. 
  } $D_{min}$ 
  which receives the allocations $\numberedAllocation{1}$ and 
  $\numberedAllocation{4}$ with probability $p_1 = 25 / 31$ and $p_4 = 6 / 31$, respectively.   
  This distribution is ex-ante proportional
  as 
  {$\mathbb{E}_{(A_1, A_2) \sim D_{min}}[v_1(A_1)] = 96 \cdot \frac{25}{31} > 77$}
  and 
  {$\mathbb{E}_{(A_1, A_2) \sim D_{min}}[v_2(A_2)]= 50$.}
  Moreover,
  \begin{align*}
    \Phi(D_{min}) &= \Var_{      {(A_1, A_2)}
      \sim D}[v_1(
        {A_1}
        )] 
        + \Var_{
          {(A_1, A_2)}
          \sim D}[v_2(
            {A_2}
            )] \\
    &= 
    \Var\left[(96 - 0) \cdot \textnormal{Ber}(25 / 31)\right] 
    + \Var\left[(100 - 38) \cdot 
    \textnormal{Ber}(6 / 31)\right] \\
    &= \frac{25}{31} \cdot \frac{6}{31}((96 - 0)^2 + (100 - 38)^2) < 2039 
  \end{align*}
  where $\textnormal{Ber}(p)$ denotes a Bernoulli random variable with parameter $p$.

  On the other hand,
  the only ex-ante proportional  
  distribution 
  over the allocations $\numberedAllocation{1}$ 
  and $\numberedAllocation{2}$ is the uniform distribution (each has probability half).
  Denote this distribution by $D_{good}$ (as it has ``good'' ex-post fairness: it is EFX and MMS ex-post).
  This distribution has a sum-of-variances value of 
  $\Phi(D_{good}) = 
  \Var\left[(96 - 4) \cdot \textnormal{Ber}(1 / 2)\right] 
  + \Var\left[(62 - 38) \cdot 
  \textnormal{Ber}(1 / 2)\right] = 
  \frac{1}{4}(92^2 + 24^2) = 2260 > 2039 > \Phi(D_{min})$. 
  Thus, every ex-ante proportional sum‐of‐variances-minimizing distribution
  $D\in \mathcal{D}$  
  must either have the allocation 
  $\numberedAllocation{3}$ or the 
  allocation $\numberedAllocation{4}$ in its
  support. As both $\numberedAllocation{3}$ and $\numberedAllocation{4}$ 
  provide one of the agents an empty bundle, the claim follows.  
\end{proof}

  { 
    We can think of the 
    distribution $D_{min}$ as starting from the 
    ``balanced'' distribution $D_{good}$ 
    (which assigns 
    one good to each agent)
     and then shifting all 
    probability mass from allocation 
    $\numberedAllocation{2}$ towards allocations 
    $\numberedAllocation{1}$ 
    and $\numberedAllocation{4}$.
    The key point is that
    allocation $\numberedAllocation{1}$ gives agent 2 
    {a bundle of reasonable value (38),} 
    while 
    allocation $\numberedAllocation{4}$ gives agent 2 a 
    bundle of maximal value (100). As $\numberedAllocation{2}$ gives agent 2 a 
    value in between (62), 
    we can shift more probability mass to $\numberedAllocation{1}$
    rather than $\numberedAllocation{4}$
    while still keeping the distribution ex-ante 
    proportional.\footnote{ 
      As allocation $\numberedAllocation{1}$ 
      assigns a large bundle (96) to agent 1, 
      shifting $D_{min}$ towards $\numberedAllocation{1}$ 
      ensures that agent 1 receives her proportional 
      share ex-ante.
    }
    This is done without greatly increasing the imbalance between the agents’ bundle values - which we call their gaps.
    Specifically, this adjustment 
    lets us choose a skewed probability  
    $p > 0.8$ such that
    the allocation $\numberedAllocation{1}$
    has probability $p$
    and $\numberedAllocation{4}$ has the remaining probability.
    As the sum-of-variances depends  
    on the sum of squared gaps,
    the maximal gap dominates the expression.
    But, the larger gap increases only slightly (from $92$ to $96$),
    and so,
    the resulting increase in sum-of-variances is minor. 
    At the same time, the probabilities become much 
    more concentrated on a single allocation ($\numberedAllocation{1}$), 
    leading to an overall reduction in sum-of-variances, and 
    $\Phi(D_{min}) / \Phi(D_{good}) < 0.91 < 1$. 
    }
  
As hinted by denoting the distribution $D_{min}$ as such,
it can be shown that this distribution is the only 
ex-ante proportional sum‐of‐variances-minimizing distribution. 
We state this formally in the following proposition.

\begin{restatable}{proposition}{nonidenticaltwosovminimizing}
    \label{prop:nonidentical2-sov-minimizing}
    For \cref{example:lower-bound} (presented in 
    \cref{prop:nonidentical2-impossibility}), the distribution
    $D_{min}$
    which 
    receives the allocations $\numberedAllocation{1}$ 
    and $\numberedAllocation{4}$ with probabilities 
    $p_1 = 25 / 31$ and $p_4 = 6 / 31$, respectively,
    is the unique 
    ex-ante proportional
    sum-of-variances-minimizing distribution
    over the allocations
    $\numberedAllocation{1}, \numberedAllocation{2}, \numberedAllocation{3}, \numberedAllocation{4}$. 
\end{restatable}

\textit{Proof deferred to \cref{sec:analyzing-sov-objective}.}

\section{Beyond Sum-of-Variances Minimization}\label{sec:beyond}

Up to now we have focused on the objective of sum-of-variances-minimization (subject to ex-ante proportionality). We have proven that such an optimization fails to ensure ex-post fairness. 
But one may wonder if selecting another objective function over the agents' variances would solve the problem. 
We next show that multiple other natural objective functions also fail to ensure ex-post fairness.  

Recall that for valuation vector $\vec{v}= (v_1,v_2, \ldots, v_n)$, we use $PROP_{\vec{v}}$ to denote the set of ex-ante proportional valuations. 
Let $\Psi_{\vec{v}}(\cdot)$ be an objective function that for a distribution $D\in PROP_{\vec{v}}$ outputs a non-negative value. Consider the problem of minimizing $\Psi_{\vec{v}}$ over the set of distributions $PROP_{\vec{v}}$, and let $\mathcal{D}_{\vec{v}}^{\Psi}$ denote the set of distributions that minimize $\Psi_{\vec{v}}$ over  $PROP_{\vec{v}}$. 
Our goal in this section is to generalize the negative result of \cref{prop:nonidentical2-impossibility} to other objective functions.

\begin{lemma}\label{lem:other-obj}
Consider an instance 
with two agents with additive valuations over two goods, 
where both agents have a non-zero valuation for each good, 
the valuations are normalized
and {there exists a good that both agents value 
strictly more than half of their total value}.
Then, there exists a  
unique ex-ante proportional distribution $D_{good}$ 
that is EFX and MMS ex-post 
({the distribution which flips a fair coin 
to decide whether agent 1 receives the first good 
and agent 2 the second, or vice versa}).
Moreover, if for objective $\Psi_{\vec{v}}$ it holds that 
$D_{good}\notin \mathcal{D}_{\vec{v}}^{\Psi}$,
then 
every distribution $D\in \mathcal{D}_{\vec{v}}^{\Psi}$ (an ex-ante proportional distribution that minimizes $\Psi_{\vec{v}}$) is supported on an allocation which provides one of the agents with an empty bundle.
Thus, every distribution $D\in \mathcal{D}_{\vec{v}}^{\Psi}$ 
is not EF1 ex-post and provides no constant-factor approximation to the MMS ex-post. 
\end{lemma}

\begin{proof}
  {
  We first prove that there exists a unique ex-ante proportional distribution $D_{good}$ that is EFX and MMS ex-post. 
  Notice that the only MMS allocations 
  are the two allocations which provide 
  one good to each agent. 
  As there exists a good which both players value 
  strictly more than half of their total value, 
  the only distribution which is  
  ex-ante proportional and ex-post MMS is $D_{good}$. 
  
  For the moreover part, 
  as there exists a good which both players value 
  strictly more than half of their total value,
  there exists no deterministic proportional allocation.
  Thus, $D_{good}$ is the sole ex-ante proportional
  distribution supported only on allocations which 
  provide both agents non-empty bundles. 
  But, as $D_{good} \not \in \mathcal{D}_{\vec{v}}^\Psi$, 
  every distribution $D\in \mathcal{D}_{\vec{v}}^\Psi$
  contains an allocation which provides one of the agents
  an empty bundle.\footnote{
  Note that for the claim to be meaningful we need to show that $\mathcal{D}^\Psi_{\vec{v}} \neq  \emptyset$.
  We later
  show that this holds 
  for every  objective we consider.}}
\end{proof}

{The main minimization objective we have studied so far was the sum-of-variances, which is equivalent to the average over the variances. Next, we consider four other objectives, based on variances and on standard-deviations:}
{
The maximal variance objective,
the ``variance of variances'' objective, 
``standard-deviation of standard-deviations'' objective
and the sum of standard-deviations objective.

{Observe that for each objective $\Psi_{\vec{v}}$ of the four, the set $\mathcal{D}^{\Psi}_{\vec{v}}$ of ex-ante proportional minimizers is not empty.
Indeed, each objective $\Psi_{\vec{v}}$ is continuous, and the set $PROP_{\vec{v}}$ is both compact and non-empty. 
Thus, the objective $\Psi_{\vec{v}}$ has a minimizer $D^\ast \in \mathcal{D}^{\Psi}_{\vec{v}}$ over the set 
$PROP_{\vec{v}}$.}

For the first three objectives, our proof will be  based on \cref{example:lower-bound}  which clearly satisfies the conditions of \cref{lem:other-obj}.
We will be using the ex-ante proportional distribution $D_{min}$ (presented in \cref{prop:nonidentical2-sov-minimizing}) to show that for each
objective $\Psi_{\vec{v}}$ of the three objectives, 
$\Psi_{\vec{v}}(D_{min})<\Psi_{\vec{v}}(D_{good})$. 
This clearly implies that $D_{good}$ does not minimize the objective $\Psi_{\vec{v}}$ over ex-ante proportional distributions, 
and we can apply the lemma to conclude that the support of any $\Psi_{\vec{v}}$ minimizing distribution is not ex-post fair.\footnote{
    Additionally, as $\mathcal{D}^\Psi_{\vec{v}} \neq \emptyset$
    there exists such distributions.
}
}

{
We next list
the variances of each agent under $D_{min}$ and $D_{good}$,
as they are useful in {computing the objectives we consider, and in proving the claim:} 
\begin{align*}
  &\Var_{(A_1, A_2)\sim D_{min}}[v_1(A_1)] = \frac{25}{31} \cdot \frac{6}{31} \cdot 96^2 < 1439,
  &\Var_{(A_1, A_2)\sim D_{min}}[v_2(A_2)] = \frac{25}{31} \cdot \frac{6}{31} \cdot 62^2 = 600, \\
  &\Var_{(A_1, A_2)\sim D_{good}}[v_1(A_1)] = \frac{1}{4} \cdot 92^2 = 2116,
  &\Var_{(A_1, A_2)\sim D_{good}}[v_2(A_2)] = \frac{1}{4} \cdot 24^2 = 144. \\
\end{align*}
Next, using \cref{lem:other-obj}, we show that 
each of the following optimization problems fails to ensure ex-post fairness: 

\begin{enumerate}

    \item \textbf{Minimizing the Largest Variance of Any Agent.}  
    The largest variance of any agent objective is defined as\footnote{ 
      As the square-root function  
      is monotonically increasing, 
      the result also holds when minimizing the largest standard deviation of any agent.
    }\textsuperscript{,}\footnote{ 
      One may also consider a lexicographic version of this objective, 
      where one first minimizes the largest variance,
      and subject to that, minimizes the second-largest variance, and so on.
      As the distribution $D_{min}$ already differs 
      from $D_{good}$ in the largest variance,
      the result also holds for this lexicographic version.
    }
    $\Psi_{\vec{v}}(D) = \max_{i\in [n]} \Var_{(A_1, A_2,\ldots, A_n)\sim D}[v_i(A_i)]$.
    Going back to \cref{example:lower-bound},
    the distribution $D_{good}$ achieves  a value
    $$\Psi_{\vec{v}}(D_{good}) = \max\{\Var_{{(A_1, A_2)}\sim D_{good}}[v_1({A_1})], \Var_{{(A_1, A_2)}\sim D_{good}}[v_2({A_2})]\} = 2116.$$
    While $D_{min}$ achieves a smaller value: 
    $$\Psi_{\vec{v}}(D_{min}) = \max\{\Var_{{(A_1, A_2)}\sim D}[v_1({A_1})], \Var_{{(A_1, A_2)}\sim D
    }[v_2({A_2})]\} < 1439< 2116 = \Psi_{\vec{v}}(D_{good}).$$
    We conclude that {$D_{good}$}   
    does not minimize this objective, and thus this objective also fails to provide EF1
    ex-post and any approximation to the MMS ex-post.
  
    \item \textbf{Minimizing the ``Variance of Variances''.}
    The ``Variance of Variances'' objective
    is a measure which aims to ``concentrate'' the variances of 
    all agents around their 
    mean.\footnote{{The \emph{expectation-of-the-variances} objective is equivalent to the sum-of-variances objective, so all our results for the later objective apply to the former. Here we focus on the variance of the variances.}}
    This can be thought of as defining the random variable 
    $X$ which receives each variance $\Var_{(A_1,A_2,\ldots, A_n)\sim D}[v_i(A_i)]$ 
    with probability $1/n$,  and then taking the variance of $X$.
    Thus, the ``Variance of Variances'' objective is defined as
    $\Psi_{\vec{v}}(D) = \frac{1}{n} \sum_{i=1}^n (\Var_{(A_1,A_2,\ldots, A_n)\sim D}[v_i(A_i)] - \mu_{var}(D))^2$,
    where $\mu_{var}(D) = \frac{1}{n} \sum_{i=1}^n \Var_{(A_1,A_2,\ldots, A_n)\sim D}[v_i(A_i)]$ is the average variance.
    Going back to \cref{example:lower-bound}, for the distribution $D_{good}$ 
    the random variable $X$ receives the values $2116$ and $144$ with probability $1/2$ each,
    i.e. $X \sim (2116 - 144) \cdot \textnormal{Ber}(1 / 2) + 144$. 
    Thus,
    $\Psi_{\vec{v}}(D_{good}) = \frac{1}{4} \left( 2116 - 144 \right)^2 = 972196$. 
    Using a similar computation, we 
    can see that for $D_{min}$ 
    we have the random variable 
    $X\sim 
    (\Var_{(A_1, A_2)\sim D_{min}}[v_1(A_1)] 
    - \Var_{(A_1, A_2)\sim D_{min}}[v_2(A_2)]) 
    \cdot \textnormal{Ber}(1 / 2) 
    + \Var_{(A_1, A_2)\sim D_{min}}[v_2(A_2)]$
    and $D_{min}$   
    achieves a smaller objective value of 
    $\Psi_{\vec{v}}(D_{min}) = 
    \frac{1}{4} (\Var_{(A_1, A_2)\sim D_{min}}[v_1(A_1)] 
    - \Var_{(A_1, A_2)\sim D_{min}}[v_2(A_2)])^2  
    < \frac{1}{4} (1439 - 600)^2 < 175981 < 972196= \Psi_{\vec{v}}(D_{good})$. 
    We conclude that {$D_{good}$}   
    does not minimize this objective, and thus this objective also fails to provide EF1
    ex-post and any constant-factor approximation to the MMS ex-post.
    \item \textbf{Minimizing the ``Standard Deviation of the Standard Deviations''.} 
    Similarly to the previous objective,
    this objective can be thought of as defining the random variable 
    $X$ which receives each standard deviation $\sigma_i(D)$ with probability $1/n$
    and then taking the standard deviation of $X$.
    Thus, the standard deviation of the standard deviations objective is defined as
    $\Psi_{\vec{v}}(D) = \sqrt{\frac{1}{n} \sum_{i=1}^n (\sigma_i(D) - \mu_{\sigma}(D))^2}$,
    where $\mu_{\sigma}(D) = \frac{1}{n} \sum_{i=1}^n \sigma_i(D)$ is
    the average standard deviation.
    Going back to \cref{example:lower-bound}, we get
    $\sigma_1(D_{good}) = 46$ and $\sigma_2(D_{good}) = 12$.
    Thus, computing {the expected value and the} 
    standard deviation 
    of the random variable $X$ which receives $\sigma_1(D_{good})$
    and $\sigma_2(D_{good})$ with probability $1/2$ each,
    we get {$\mu_{\sigma}(D) = (46+12)/2=29$ and
    $\Psi_{\vec{v}}(D_{good}) 
    = \sqrt{\frac{1}{2}((46 - 29)^2 + (12 - 29)^2)} = 17$, respectively.}
    On the other hand,
    $\sigma_1(D_{min}) = \sqrt{\Var_{(A_1, A_2)\sim D_{min}}[v_1(A_1)]} < {\sqrt{1439}}
    < 38$ and $\sigma_2(D_{min}) = \sqrt{\Var_{(A_1, A_2)\sim D_{min}}[v_2(A_2)]} {=\sqrt{600}} > 24$,
    and thus  
    {
    \begin{align*}
      \Psi_{\vec{v}}(D_{min}) 
      &= \sqrt{\frac{1}{2}((\sigma_1(D_{min}) - \mu_{\sigma}(D))^2 + (\sigma_2(D_{min}) - \mu_{\sigma}(D))^2)} \\
      &= \sqrt{\frac{1}{4}(\sigma_1(D_{min}) - \sigma_2(D_{min}))^2} \\ 
      &< \sqrt{\frac{1}{4}(38 - 24)^2} = 7 < 17 = \Psi_{\vec{v}}(D_{good}).
    \end{align*}}
    We conclude that {$D_{good}$}  
    does not minimize this objective, and thus this objective also fails to provide EF1
    ex-post and any approximation to the MMS ex-post.
    
    \item {\textbf{Minimizing the Sum-of-Standard-Deviations.}}
    The sum-of-standard-deviations objective is defined 
    as $\Psi_{\vec{v}}(D) = \sum_{i=1}^n \sigma_i(D)$, where
    $\sigma_i(D) = \sqrt{\Var_{(A_1, A_2,\ldots, A_n)\sim D}[v_i(A_i)]}$ 
    is the standard deviation of agent $i$.
    To show the failure 
    of the sum-of-standard-deviations 
    fairness ex-post,
    we provide the 
    following example:\footnote{We present this additional example as \cref{example:lower-bound} 
    does not serve as a counterexample for this objective. Indeed, in 
    \cref{example:lower-bound}, the distribution 
    $D_{good}$ is an ex-ante proportional 
    sum-of-standard-deviations-minimizing distribution.}

    \begin{example}\label{example:lower-bound-std}
    Consider a setting with two agents with the following additive valuations over the set $\mathcal{M} =\{a,b\}$ of two goods:
      \begin{align*}
        &v_1(a)=90, v_1(b)=10; \\
        &v_2(a)=55, v_2(b)=45.
      \end{align*}    
      In this setting there are four possible allocations of $\mathcal{M}$ into two bundles, 
      we enumerate them with the corresponding value-vectors:
      \begin{align*}
        \numberedAllocation{1} &=(\{a\},\{b\})\to(90,45), &  \numberedAllocation{2} &=(\{b\},\{a\})\to(10,55), \\
        \numberedAllocation{3} &=(\{a,b\},\emptyset)\to(100,0), &  \numberedAllocation{4} &=(\emptyset,\{a,b\}) \to(0,100). 
      \end{align*}
    \end{example}

  Note that the valuations are normalized  ($v_1(\mathcal{M})=v_2(\mathcal{M})=100$),
  both agents have a non-zero valuation for each good,
  and there exists an item ($a$) which 
  both agents value strictly more than 
  half of their total value 
  ($v_1(a) > v_1(\mathcal{M}) / 2$, $v_2(a) > v_2(\mathcal{M}) / 2$). 
  Thus, the conditions of \cref{lem:other-obj}
  are met by \cref{example:lower-bound-std}.

  Consider the distribution $F$  
  which receives the allocations 
  $\numberedAllocation{1}$ and 
  $\numberedAllocation{4}$ with 
  probability $p_1 = 10 / 11$ 
  and $p_4 = 1 / 11$, respectively.   
  This distribution is ex-ante proportional
  as 
  {$\mathbb{E}_{(A_1, A_2) \sim F}[v_1(A_1)] = 90 \cdot \frac{10}{11} > 81$}
  and 
  {$\mathbb{E}_{(A_1, A_2) \sim F}[v_2(A_2)]= 50$.}
  Moreover,
  \begin{align*}
    \Phi_{\vec{v}}(F) &= \sigma_1(F) + \sigma_2(F) \\ 
    &= \sqrt{\Var_{{(A_1, A_2)}\sim F}[v_1({A_1})]}
    + \sqrt{\Var_{{(A_1, A_2)}\sim F}[v_2({A_2})]} \\ 
    &= \sqrt{\Var[(90 - 0) \cdot \textnormal{Ber}(10 / 11)]}
    + \sqrt{\Var[(45 - 100) \cdot \textnormal{Ber}(10 / 11)]} \\
    &= \sqrt{\frac{1}{11} \cdot \frac{10}{11} \cdot 90^2}
    + \sqrt{\frac{1}{11} \cdot \frac{10}{11} \cdot 55^2} 
    < 42.
  \end{align*}
  where $\textnormal{Ber}(p)$ denotes a Bernoulli random variable with parameter $p$. 
  On the other hand,
  the only ex-ante proportional  
  distribution 
  over the allocations $\numberedAllocation{1}$ 
  and $\numberedAllocation{2}$ is the uniform distribution (each has probability half).
  Denote this distribution by $F_{good}$.
  This distribution has a sum-of-standard-deviations
  value of
  \begin{align*}
    \Phi(F_{good}) 
    &= \sigma_1(F_{good}) + \sigma_2(F_{good}) \\
    &= \sqrt{\Var[(90 - 10) \cdot \textnormal{Ber}(1 / 2)]}
    + \sqrt{\Var[(45 - 55) \cdot \textnormal{Ber}(1 / 2)]} \\
    &= \sqrt{\frac{1}{4} \cdot 80^2} + \sqrt{\frac{1}{4} \cdot 10^2} \\ 
    & = 45 > 42 = \Phi(F)
  \end{align*}
  {We thus see that $F_{good}$, the only ex-ante proportional distribution 
  for which 
  both $P_3$ and $P_4$ are not in the support,  
  is not a sum‐of‐standard-deviations-minimizing distribution, so
  $F_{good} \not \in \mathcal{D}_{\vec{v}}^\Psi$.} 
  By \cref{lem:other-obj} 
  we conclude that
  this objective fails
  to provide ex-post fairness as well.  
\end{enumerate}

\section{Conclusion}\label{sec:conclusion}

In this paper, we considered the intuitive approach 
of ``minimizing variance'' by considering 
distributions which explicitly minimize the
sum-of-variances of the agents, subject to 
ex-ante proportionality. 
When agents have identical additive valuations,
every allocation in the support is EFX, and therefore
provides every agent a constant fraction of her maximin share (but if there are more than two agents it does not guarantee the full MMS).
On the other hand, in stark contrast
to the suggested intuition, 
when agents have non-identical 
additive valuations, 
we demonstrated that such sum-of-variances
minimization may entirely fail to 
guarantee any meaningful ex-post fairness. Even in the minimal case 
of two agents and two goods, the resulting distribution can assign an 
empty bundle to an agent, violating EF1 and not providing any constant-factor approximation to the MMS.
This finding underscores the limits of total variance 
minimization as a tool for designing Best-of-Both-Worlds 
fair distributions.
We have also shown that similar negative results hold 
for other variances-based and standard-deviations-based 
minimization objectives. We leave open the problem of 
finding some natural variances-based or 
standard-deviations-based objective  for which 
positive results can be proven.

\section*{Acknowledgments}
This research was supported by the Israel Science Foundation (grant No. 301/24).
Moshe Babaioff's research is also supported by a Golda Meir Fellowship.

\bibliographystyle{plain}
\bibliography{references}

\newpage

\appendix

\section{Missing Proofs for Identical Agents}
\label{app:proofs}

The following simple claim will be useful in the proof of \cref{lem:diff-minimizing}.

\begin{restatable}[Sum‐of‐variances Formula]{lemma}{sumvarianceformula}
\label{lem:sum-variance-formula}
Let $D$ be a distribution over allocations 
of $\mathcal{M}$ 
to two agents
and $v$ be an additive valuation over $\mathcal{M}$. 
Then
\[ \sum_{i=1}^{2} \Var_{(A_1, A_2) \sim D}[v(A_i)] 
= 2\cdot \Var_{(A_1, A_2) \sim D}[v(A_1)] = 2\cdot \Var_{(A_1, A_2) \sim D}[v(A_2)]. \]
\end{restatable}

\begin{proof}
Since $v(A_1)+v(A_2)= v(\mathcal{M})$ is a constant, 
$\Var_{(A_1,A_2)\sim D}[v(A_1)] = \Var_{(A_1,A_2)\sim D}[v(\mathcal{M}) - v(A_2)] = \Var_{(A_1,A_2)\sim D}[v(A_2)]$. Therefore,
\begin{align*}
  \sum_{i=1}^{2} \Var_{(A_1,A_2)\sim D}[v(A_i)] 
  &= \Var_{(A_1,A_2)\sim D}[v(A_1)] + \Var_{(A_1,A_2)\sim D}[v(A_2)] 
  = 2\,\Var_{(A_1,A_2)\sim D}[v(A_1)].
\end{align*}
An analogous equality holds for $v(A_2)$.
\end{proof}

\suppdiffmin*

\begin{proof} 
Fix $D\in \mathcal{D}$ to be some ex-ante proportional 
sum-of-variances-minimizing distribution. 
Assume that $D$ has support of size $k$, and recall that 
{
{$\numberedAllocation{j} = 
\left(\numberedBundle{j}{1}, \numberedBundle{j}{2}\right)$} }
denotes the $j$-th allocation in the support 
$ \textnormal{Supp}(D) = \left\{ 
  \numberedAllocation{j}
 \right\}_{j=1}^{k}$, 
and that we denote the probability that $D$ assigns to 
allocation  
{$\numberedAllocation{j}$}
by $p_j$.
Suppose, for contradiction, that $D$ assigns positive 
probability to an 
allocation 
{$\numberedAllocation{j{^\prime}}=
(\numberedBundle{j{^\prime}}{1}, \numberedBundle{j{^\prime}}{2})$} 
for some $1\leq j^{\prime}\leq k$ 
that is not difference minimizing. 
Thus, 
for some difference‐minimizing allocation $S=(S_1,S_2)$,
we have that
{
$|v( \numberedBundle{j{^\prime}}{1})-v( \numberedBundle{j{^\prime}}{2})| 
> |v(S_1)-v(S_2)|$}. 
Construct a distribution $D'$ that randomizes uniformly over $S=(S_1, S_2)$ and its permutation $(S_2,S_1)$.
Clearly, $D'$ is ex-ante proportional.
Computing the variance of agent $1$ we have that: 
  \begin{align*}
    \textnormal{Var}_{(A_1,A_2)\sim D}[v(A_1)]
    &= \sum_{j=1}^{k} p_j \cdot 
    \left(v\left(\numberedBundle{j}{1}\right) - \textnormal{PS}\right)^2 \\
    &= \sum_{j=1}^{k} p_j \cdot \left(v\left(\numberedBundle{j}{1}\right) - \frac{1}{2}\left[v\left(\numberedBundle{j}{1}\right) + v\left(\numberedBundle{j}{2}\right)\right]\right)^2 \\
    &= \frac{1}{4} \sum_{j=1}^{k} p_j \cdot 
    \left(v\left(\numberedBundle{j}{1}\right) 
    - v\left(\numberedBundle{j}{2}\right)\right)^2 \\
    &> \frac{1}{4} \sum_{j=1}^{k} p_j \cdot \left(v\left(S_1\right) - v\left(S_2\right)\right)^2 \\
    &= \frac{1}{4} \left(v\left(S_1\right) - v\left(S_2\right)\right)^2 \\
    &= \textnormal{Var}_{(A_1,A_2)\sim D'}[v(A_1)].
  \end{align*}
  Thus, by \cref{lem:sum-variance-formula}, 
  $\Phi(D) =  
  2  \cdot \Var_{(A_1,A_2)\sim D}[v(A_1)] > 2 \cdot \Var_{(A_1,A_2)\sim D'}[v(A_1)] 
  = \Phi(D')$, 
  contradicting the assumption that $D$ is sum-of-variances-minimizing 
  over the set of ex-ante proportional distributions for $(v,v)$.
\end{proof}

\diffminimizingmmsfair*

{This simple claim is well known \cite{Korf10}, we present its proof for completeness.}

\begin{proof}
  Let $S$ be a difference‐minimizing allocation of $\mathcal{M}$ 
  and assume w.l.o.g that $v(S_1) \leq v(S_2)$.
  Assume for contradiction that $v(S_1) < \MMS$.
  Now, let $P = (B_1, B_2)$ be an MMS-achieving allocation.
  Assume w.l.o.g that $v(B_1) \leq v(B_2)$.
  So $ \textnormal{MMS} = v(B_1)$.
  As $v(S_1) \leq v(S_2)$ and $S$ is not MMS-fair   
  we have that $v(S_1) < \textnormal{MMS} = v(B_1)$. 
  As $v(B_1) + v(B_2) = v(S_1) + v(S_2) = v(\mathcal{M})$, 
  we have that $v(S_2) > v(B_2) > v(B_1) > v(S_1)$. 
  Thus, $ v(S_2) - v(S_1) > v(B_2) - v(B_1)$
  in contradiction to the fact that $S$ is difference minimizing.
\end{proof}

\section{Missing Proofs for Two Non-Identical Agents}
\label{app:non-identical2-missing-proof}

\subsection{Analyzing the SoV Objective}
\label{sec:analyzing-sov-objective}

In this section we consider \cref{example:lower-bound} (presented in \cref{prop:nonidentical2-impossibility}) and show that the distribution $D_{min}$ presented in the proof
is indeed the only ex-ante proportional sum‐of‐variances-minimizing distribution.
We do so using standard tools from linear programming and convex analysis.

\begin{definition}
Let $\Delta^3$ be the space of all distributions  over $\numberedAllocation{1}, \numberedAllocation{2}, \numberedAllocation{3}, \numberedAllocation{4}$.
\[
\Delta^3 = \left\{ (p_1,p_2,p_3,p_4) \in \mathbb{R}^4 \;\middle|\; 
p_j \geq 0 \;\;\forall j\in [4],\;\; \sum_{j=1}^4 p_j = 1 \right\}.
\]

  Let $\mathcal{P}$ be the set of all possible distributions over 
  $\numberedAllocation{1}, \numberedAllocation{2}, \numberedAllocation{3}, \numberedAllocation{4}$
  which provide both agents at least 
  their proportional share ex-ante. 
  Formally, 
  \begin{align*}
    \mathcal{P} = \left\{ \vb{p} \in \Delta^3\ \middle|\  
    \sum_{j=1}^4 p_j \cdot v_1(\numberedBundle{j}{1}) \geq 50\ and\ 
    \sum_{j=1}^4 p_j \cdot v_2(\numberedBundle{j}{2}) \geq 50 \right\}.
  \end{align*}
  We think of $\vb{p}$ as the distribution which 
  receives allocations 
  $ \numberedAllocation{j} = 
  \left(\numberedBundle{j}{1}, \numberedBundle{j}{2}\right)$ with probability $p_j$. 

\end{definition}

\begin{restatable}{lemma}{phiconcave}
  \label{lem:phi-concave}
  The sum-of-variances function
  $\Phi: \mathcal{P} \to \mathbb{R}_{\geq 0}$ 
  is strongly concave and continuous.
\end{restatable}

\textit{Proof deferred to \cref{sec:missing-proofs-analyzing-sov}.}

\begin{restatable}{lemma}{compactconvex}
  \label{lem:compact-convex}
  The set $\mathcal{P}$ of ex-ante proportional distributions
  is compact and convex.  
\end{restatable}

\textit{Proof deferred to \cref{sec:missing-proofs-analyzing-sov}.}

\begin{definition}
    Let $\mathcal{C} \subseteq\mathbb{R}^n$
    be a convex set.
   An \emph{extreme point} $x\in \mathcal{C}$ 
   of $\mathcal{C}$ is a point such that 
   if $x = \lambda y + (1 - \lambda) z$ for some $y, z \in \mathcal{C}$ and $\lambda \in (0,1)$,
   then $y = z = x$.
\end{definition}

The following is a classic result 
from the theory of convex analysis.
It can be found in classic textbooks such as \cite{Rockafellar1970}. 

\begin{theorem}
  \label{thm:convex-convave-min}
  For a continuous and strongly concave function $f$ over a compact and convex set $\mathcal{C}$
  the global minimum will be attained at an extreme point of $\mathcal{C}$
  and no interior point of $\mathcal{C}$ can be a global minimum of $f$.
\end{theorem}

\begin{definition}
    Consider a polyhedron $\mathcal{H} = \left\{ x \in \mathbb{R}^n | Ax \leq b \right\}$ in $\mathbb{R}^n$.
    A point $\vb{p} \in \mathcal{H}$ is called a \emph{vertex of $\mathcal{H}$} if 
    $\vb{p}$ is the solution to $n$ linearly independent constraints of $\mathcal{H}$.  
\end{definition}

A well-known result used in the field 
of linear programming is the following theorem.

\begin{theorem}
    \label{thm:extreme-vertex}
    Let $\mathcal{H}$ be a polyhedron in $\mathbb{R}^n$.
    Then, a point $\vb{p} \in \mathcal{H}$ is an extreme point of $\mathcal{H}$ if and only if
    it is a vertex of $\mathcal{H}$.
\end{theorem}

\begin{restatable}{corollary}{extremepoints}
  \label{cor:extreme-points}
  The extreme points of $\mathcal{P}$ satisfy one of the following: 
  \begin{itemize}
    \item At least one of the ex-ante proportional share constraints is tight
    and at least two of the $p_i$'s are zero. Thus, at most 
    two $p_i$'s are non-zero. 
    \item Both ex-ante proportional share constraints are 
    tight and at least one of the $p_i$'s is zero. 
    Thus, at most three $p_i$'s are non-zero.
  \end{itemize}
\end{restatable}

\textit{Proof deferred to \cref{sec:missing-proofs-analyzing-sov}.}

\nonidenticaltwosovminimizing*

\begin{proof}
  Due to \cref{thm:convex-convave-min},
  it is enough to show that all the extreme points $\vb{p}$ 
  of $\mathcal{P}$
  except $D_{min}$ 
  satisfy $\Phi(\vb{p}) > \Phi(D_{min})$. 
  Moreover, due to \cref{cor:extreme-points}, 
  we know that all extreme points of $\mathcal{P}$
  satisfy one of two structural cases presented in the corollary.
  We begin with the first case, 
  where at least two of the  
  constraints $p_j \geq 0$ for $j\in [4]$ are tight
  and one of the ex-ante proportional share constraints 
  is tight.
  We divide into four additional sub-cases 
  depending on which 
  of the constraints $p_j \geq 0$ for $j\in [4]$ 
  are tight.  
  {
  For the analysis of these sub-cases, 
  recall that $\Phi(D_{min}) < 2039$.} 
  \begin{itemize}
    \item  \textbf{If $p_2 = p_4 = 0$ or $p_1 = p_3 = 0$:} 
           In these two sub-cases, either only $p_1,p_3 \geq 0$, or 
           only $p_2,p_4 \geq 0$, respectively.
           Thus, in each of those sub-cases one of the agents does not receive her proportional share ex-ante.
           We conclude that there exists no extreme point 
           $\vb{p} \in \mathcal{P}$
           which satisfies 
           $p_2 = p_4 = 0$ or $p_1 = p_3 = 0$
           and is ex-ante proportional.
    \item \textbf{If $p_3 = p_4 = 0$:}
          Then, only $p_1,p_2$ can be non-zero, 
          and as each agent must receive 
          her proportional share ex-ante
          it must hold that
          $p_1 = p_2 = 1 / 2$.
          The sum-of-variances value is 
          $\Phi((1 / 2, 1 / 2, 0,0)) = 
          \Var\left[(96 - 4) 
          \cdot \textnormal{Ber}(1 / 2)\right]
          + \Var\left[(38 - 62) \cdot 
          \textnormal{Ber}(1 / 2)\right] = 
          \frac{1}{4}(92^2 + 24^2) = 2260 > \Phi(D_{min})$. 
           We conclude that all extreme points of 
           $\vb{p} \in \mathcal{P}$
           which satisfy $p_3 = p_4 = 0$
           achieve a greater sum-of-variances value than 
           $D_{min}$.
      \item {\textbf{If $p_2 = p_3 = 0$:}} 
            Then only $p_1,p_4$
            can be non-zero.
            Demanding that one of the ex-ante 
            proportional share constraints is tight,
            we have that either 
            $p_1 = 50 / 96$ or $p_1 = 25 / 31$.
            The latter is the distribution $D_{min}$,
            and it obviously achieves a smaller 
            sum-of-variances value than the former.
            We conclude that the only extreme point $\vb{p} \in \mathcal{P}\setminus \{D_{min}\}$ that satisfies $p_2 = p_3 = 0$ is one that  
            has a greater sum-of-variances value than 
            $D_{min}$.
      \item 
            {\textbf{If $p_1 = p_4 = 0$:}} 
            Then only $p_2,p_3$
            can be non-zero.
            By demanding 
            that one of the ex-ante proportional share constraints is tight,
            we have that $p_2 = 50 / 62$ or 
            $p_2 = 50 / 96$.
            If $p_2 = 50 / 62$, then the first agent doesn't receive her proportional share ex-ante,
            and for $p_1 = 50 / 96$ a sum-of-variances 
            value of
            $\Phi((0, 50 / 96, 46 / 96, 0)) = 
            \frac{50}{96} \cdot \frac{46}{96}((100 - 4)^2 + (62 - 0)^2) 
            > 3258 > \Phi(D_{min})$ is achieved. 
            {
              We conclude that all 
              extreme points of $\vb{p} \in \mathcal{P}$ 
              which satisfy
              $p_1 = p_4 = 0$
              achieve a 
              greater sum-of-variances value than $D_{min}$.}
        \item {\textbf{If $p_1 = p_2 = 0$:}  
              Then only $p_3,p_4$
              can be non-zero. 
              As each agent must receive 
              her proportional share ex-ante,
              we have that $p_3 = p_4 = 1 / 2$ and the sum-of-variances value is
              $\Phi((0, 0, 1 / 2, 1 / 2)) = \frac{1}{4} \cdot (100^2 + 100^2) = 5000 > \Phi(D_{min})$.
              We conclude that all extreme points of 
              $\vb{p} \in \mathcal{P}$ which satisfy 
              $p_1 = p_2 = 0$
              achieve a greater sum-of-variances value 
              than $D_{min}$.
              }
  \end{itemize}
  We proceed to the second case,
  where both ex-ante proportional share constraints are tight
  and at least one of the constraints 
  $p_j \geq 0$ for $j\in [4]$ is tight.
  Firstly, note that $(P_1, P_2)$ and 
  $(P_3, P_4)$ are \textit{completing pairs} of allocations
  in the sense that $P_1 + P_2 = P_3 + P_4 = (100,100)$.     
  Thus, setting $p_1 = p_2 = 1 / 2$
  or $p_3 = p_4 = 1 / 2$
  will provide both agents exactly their proportional share ex-ante.  
  Now, we divide into four sub-cases, 
  depending on which of the constraints
  $p_j \geq 0$ for $j\in [4]$
  is tight.{
  Let \emph{sub-case-$j$} for $j\in[4]$  be the case where $p_j = 0$. 
  Utilizing the fact that
  both ex-ante proportional share constraints are tight,
  for sub-case-$j$, 
  we derive a system of equations 
  for the variables
  $\vb{p} = (p_1, \ldots, p_4)$
  of the form $\sum_{l = 1}^4 p_{l} = 1, 
  \sum_{l = 1}^4 p_{l} 
  \cdot (v_1(\numberedBundle{l}{1}), 
  v_2(\numberedBundle{l}{2}))
  = (50,50)$,
  where $p_j = 0$.
  But, each sub-case contains exactly one 
  of the completing pairs of allocations\footnote{ 
    For $j = 1$ and $j = 2$ the completing pair is $(P_3, P_4)$, 
    while for $j = 3$ and $j = 4$ the completing pair is $(P_1, P_2)$.
  },
  and thus sub-case $j$'s system of equations 
  contains a solution of the form 
  $\vb{p} = (1 / 2, 1 / 2, 0, 0)$ or
  $\vb{p} = (0, 0, 1 / 2, 1 / 2)$.
  But, each sub-case's system of equations
  is linearly independent and thus has a unique solution.
  Thus, the only possible 
  extreme points $\vb{p} \in \mathcal{P}$ 
  which make both ex-ante proportional share constraints tight
  and at least one of the $p_j$'s zero
  are $(1 / 2, 1 / 2, 0,0)$
  and $(0, 0, 1 / 2, 1 / 2)$, 
  But, we have already shown that both these distributions
  have a sum-of-variances value larger than $\Phi(D_{min})$.
  Thus, overall we conclude that all extreme points $\vb{p} \in \mathcal{P}$ 
  except $D_{min}$ satisfy $\Phi(\vb{p}) > \Phi(D_{min})$.}
\end{proof}

\subsection{Missing Proofs For Analyzing the SoV Objective}
\label{sec:missing-proofs-analyzing-sov}

We introduce the following notations that will be useful in the proofs: 
 $(x_1, \ldots, x_n)^\intercal$ denotes the column vector 
obtained by transposing the row vector
$(x_1, \ldots, x_n)$.
For two vectors $\vb{x}$ and $\vb{y}$
of dimension $n$,
the expression $\vb{x}^\intercal \vb{y}$  
denotes their standard \emph{inner product}, i.e., the scalar
$\sum_{i = 1}^n x_iy_i$, and the expression $x \ast y$
denotes their \emph{Hadamard product}, 
a vector of dimension $n$ satisfying $(x \ast y)_i = x_i \cdot y_i$
for every $i \in [n]$.} 
We say that a matrix $M \in \mathbb{R}^{n \times n}$ is
\emph{positive-definite} if for all 
vectors $x \in \mathbb{R}^n{\setminus\{0\}}$,
$x^\intercal M x > 0$.
A matrix $M \in \mathbb{R}^{n \times n}$ is
\emph{negative-definite} if $-M$ is positive-definite.

Moreover, for a function $f: \mathbb{R}^n \to \mathbb{R}$,
the notation $\nabla^2 f$ denotes the \emph{Hessian matrix} of $f$,
satisfying $ \left(\nabla^2 f\right)_{ij} = \left( \frac{\partial^2 f}{\partial x_i \partial x_j} \right)$
for every $i,j \in [n]$.
A function $f:\mathbb{R}^n \to \mathbb{R}$
is called \emph{strongly concave} if there exists a constant $\mu > 0$ such that for all 
$x,y \in \mathbb{R}^n$ and $\lambda \in [0,1]$,
$f(\lambda x + (1-\lambda)y) \;\;\geq\;\; \lambda f(x) + (1-\lambda) f(y) + \tfrac{\mu}{2}\lambda(1-\lambda)\|x-y\|^2$.
Recall that a function $f:\mathbb{R}^n \to \mathbb{R}$ is \emph{strongly concave} 
if and only if its Hessian matrix is negative-definite everywhere.

\phiconcave*

\begin{proof}
{
    We begin by rewriting $\Phi$,
    \begin{align*}
      \Phi(\vb{p}) 
      &= \Var_{(A_1, A_2) \sim D}
      [v_1(A_1)] + \Var_{(A_1, A_2) \sim D}[v_2(A_2)] \\
      &= \mathbb{E}_{(A_1, A_2) \sim D}
      \left[v_1(A_1)^2\right] + 
      \mathbb{E}_{(A_1, A_2) \sim D}
      \left[v_2(A_2)^2\right] \\ 
      &- 
      \left(\mathbb{E}_{(A_1, A_2) \sim D}
      [v_1(A_1)]\right)^2 
      - \left(\mathbb{E}_{(A_1, A_2) \sim D}
      [v_2(A_2)]\right)^2 \\ 
      &= 
      \sum_{i=1}^{6} p_i \cdot v_1(\numberedBundle{i}{1})^2
      + \sum_{i=1}^{6} p_i \cdot v_2(\numberedBundle{i}{2})^2
      - \left(\sum_{i=1}^{6} p_i \cdot v_1(\numberedBundle{i}{1})\right)^2
      - \left(\sum_{i=1}^{6} p_i \cdot v_2(\numberedBundle{i}{2})\right)^2 \\
      &= {\vb{p}}^{\intercal}\vb{s} - ({\vb{p}}^\intercal {\vb{x}}_1)^2 - ({\vb{p}}^\intercal {\vb{x}}_2)^2,
    \end{align*}
    where $\vb{p} = (p_1, \ldots, p_4)^{\intercal}, 
    {\vb{x}}_1 = (v_{1}(\numberedBundle{1}{1}), \ldots, v_1(\numberedBundle{4}{1}))^{\intercal}, 
    {\vb{x}}_2 = (v_{2}(\numberedBundle{1}{2}), \ldots, v_1(\numberedBundle{4}{2}))^{\intercal}$ 
    and 
    ${\vb{s}} = \vb{x}_1 \ast \vb{x}_1 + \vb{x}_2 \ast \vb{x}_2$.}
    It is not difficult to see that 
    $ \nabla^2 (\vb{p}^\intercal \vb{x})^2 = 2 \vb{x} \vb{x}^\intercal$
    for any $\vb{x} \in \mathbb{R}^n$
    {
    and $ \nabla^2 ( \vb{p}^\intercal \vb{s}) = 0$.}
    But, $2 \vb{x}_1 \vb{x}_1^\intercal$ and $2 \vb{x}_2 \vb{x}_2^\intercal$ are both 
    positive-definite and thus so is their sum.
    Thus, 
    $\nabla^2 \Phi(p) = -2 (\vb{x}_1 \vb{x}^\intercal + \vb{x}_2 \vb{x}_2^\intercal)$ is negative-definite
    and so $\Phi$ is strongly concave.
    Continuity follows from the fact that $\Phi$ is polynomial.
\end{proof}

\compactconvex*

\begin{proof}
    The set $\mathcal{P}$ satisfies:
    \begin{align*}
      \mathcal{P} &= \Delta^3 \cap \left\{ \vb{p} \in \mathbb{R}^4 : \sum_{j=1}^4 p_{j} \cdot v_1(\numberedBundle{j}{1}) \geq 50 \right\} \cap
        \left\{ \vb{p} \in \mathbb{R}^4 : \sum_{j=1}^4  p_{j} \cdot v_2(\numberedBundle{j}{2}) \geq 50 \right\}.
    \end{align*}
    The set $\Delta^3$ is closed as the intersection of closed-halfspaces
    and the pre-image of $ \left\{ 1 \right\}$ 
    with respect to
    the continuous function $f(\vb{p}) = \sum_{j=1}^4 p_j$.
    The two other sets are closed as they are the pre-images of $[50, \infty)$ 
    of the continuous functions $f_1(\vb{p}) = \sum_{j=1}^4 p_{j} \cdot v_1\left(\numberedBundle{j}{1}\right)$ and
    $f_2(\vb{p}) = \sum_{j=1}^4 p_{j} \cdot v_2\left(\numberedBundle{j}{2}\right)$, respectively. 
    Thus, $\mathcal{P}$ is the intersection of a finite number of closed sets
    and thus is closed. 
   {
    Additionally, $\mathcal{P}$ 
    is clearly bounded as $\mathcal{P} \subseteq \Delta^3$
    and $\Delta^3$ is bounded.} Thus, by the Heine-Borel theorem, the set $\mathcal{P}$ is compact. 
    For convexity, notice that all constraints are linear.
\end{proof}

\extremepoints*

\begin{proof}
    It is not difficult to see that the set $\mathcal{P}$ is a polyhedron in $\mathbb{R}^4$
    and that all the constraints are linearly independent.
    Due to \cref{thm:extreme-vertex}
    we know that the extreme points of $\mathcal{P}$ are the vertices of $\mathcal{P}$.
    Thus, the extreme points are the points which make exactly four of the linear constraints tight.
    The constraint $\sum_j p_j = 1$ is always tight,
    and thus three remain. 
    All that is left is to show that 
    it is not possible that 
    three of the $p_j$'s are zero. 
    In this case, only one $p_j$ is non-zero 
    and thus must be equal to $1$. 
    But, no single allocation provides 
    the proportional share to both agents.  
\end{proof}
\end{document}

%% file: macros.tex
\newcommand{\MMS}{\mathrm{MMS}}
\newcommand{\Var}{\mathrm{Var}}

\newcommand{\numberedAllocation}[1]{P^{(#1)}}
\newcommand{\numberedBundle}[2]{P^{(#1)}_{#2}}

%% file: intro.tex
\section{Introduction}


{A fundamental problem in fair division is the problem of fairly partitioning a set $\mathcal{M}$ of indivisible goods 
among a set $\mathcal{N}$ of $n$ equally-entitled agents, when each agent $i$ has an additive valuation function $v_i:2^{\mathcal{M}} \to \mathbb{R}_{\ge0}$ over the goods.
Two central fairness notions are envy-freeness and proportionality: 
An allocation is \emph{envy-free} if each agent prefers her own bundle to the bundle of any other agent, and it is \emph{proportional} if each agent values her bundle at least at $1/n$-fraction of her value for the set of all goods.  
When goods are divisible, envy-free allocations always exist, and thus proportional allocations also exist.
Unfortunately, when goods are indivisible, proportional allocations may not exist, as evident from simple settings with two agents and an odd number of identical goods. 
}

The literature on fair allocation of indivisible goods addresses the indivisibility problem by defining  
several relaxations of the fairness notions of envy-freeness and proportionality. 
Relaxations of envy-freeness include 
envy-free-up-to-one-good (EF1) and 
envy-free-up-to-any-good (EFX). 
The notion of EF1~\cite{LiptonMMS04,Budish11}, 
requires envy to be eliminated by removing one good 
from the envied bundle. 
The stronger notion of EFX
\cite{caragiannis2019unreasonable} requires envy to be eliminated whenever any (positive-valued) good is removed from the envied bundle.
 Relaxations of the proportional share include the celebrated \emph{maximin‐share (MMS)},
 introduced by Budish~\cite{Budish11}. The MMS of an agent is the maximum value that she can guarantee herself by partitioning the goods to $n$ bundles, knowing that she will get the least valued one. 
 The MMS is an upper bound on the value that the worse-off agent will get when all agents have identical valuations. 
Ideally, every agent should get at least her MMS, but even in settings with only three agents with (different) additive valuations, an allocation that gives every agent her MMS might not exist~\cite{ProcacciaWang14} (although it is guaranteed to exist when all valuations are identical).  We further discuss these fairness notions in \cref{sec:related}.
 
Consider a setting with two agents, Alice and Bob, and 11 identical goods.  
An allocation that gives 5 goods to Alice and 6 to Bob 
is EFX (thus EF1) and also gives every agent her MMS, and thus ex-post fair. 
Yet, although Alice's ex-post allocation is fair, she may still argue that she was treated unfairly. 
To take this problem to the extreme, 
consider the case where there is only a single good, rather than 11. 
All three ex-post fairness notions of EFX, EF1 and MMS, are satisfied by deterministically giving the good to Bob. Yet,  
Alice would argue that this is unfair
as she never receives the good; Instead, 
it would be more fair to flip an unbiased coin to decide who will get the good.
Compared with the deterministic allocation 
that assigns
the good to Bob, 
this procedure has the additional advantage of 
giving each agent her proportional share in 
expectation, {and thus is fair ex-ante}.

Obtaining ex-ante proportionality is trivial regardless of the number of goods: simply allocate all goods to a random agent. 
But would agents consider this randomized procedure fair? 
An agent who ends up with nothing when the other gets all goods, might argue that it is not.
Going back to the example with two agents and 11 identical goods, an alternative procedure is to uniformly randomize over which agent gets 5 goods and which gets 6. Not only is this distribution ex-ante proportional, 
but additionally,
every allocation in the support is also fair ex-post, in that it is EFX and guarantees each agent her MMS. 
{Best-of-Both-Worlds distributions such as these obtain desirable fairness notions both ex-ante and ex-post}. 
Multiple recent papers have studied Best-of-Both-Worlds procedures under various fairness notions 
\cite{Aziz2019probabilistic, Freeman2020, Aziz2020, Babaioff2022, Hoefer2022, AzizGG23, FeldmanMNP24,
AleksandrovAGW15, HalpernPPPS20, BabaioffEF21}, we discuss them in \cref{sec:related}.

But why would agents find the uniform distribution over 5 and 6 goods superior to the  uniform distribution over 11 and 0, although both give the same expected value? One intuitive reason is that for the first distribution, each agent has lower variance for her value of the outcome, giving her greater predictability of her final value.  
This suggests a general approach for picking distributions that are fair: pick
a distribution that ``minimizes variances'', subject to being ex-ante proportional. In this paper we aim to formalize this approach and study the ex-post fairness properties of the resulting distributions.

\subsection{Sum-of-Variances Minimization and Our Contributions}
We are interested in formalizing the intuitive approach of picking 
a distribution that ``minimizes variances'', subject to being ex-ante proportional. Given a distribution over allocations, each of the agents has her own variance of her valuation. How should the variances of the different agents be aggregated?  Our main focus in this paper is 
on minimizing the sum-of-variances of the values of the agents\footnote{Which is equivalent to minimizing the average of the variances.}, subject to the distribution being ex-ante proportional.
(In  \cref{sec:beyond} we discuss other objectives).\footnote{Note that the alternative objective of  \emph{minimizing the variance of the sum of values over ex-ante proportional distributions} completely fails to give any ex-post fairness. When all valuations are identical, the sum is constant, and thus any allocation can appear in the support, including giving all goods to a random agent{, which is ex-post unfair}.} 
In support of this approach, observe that such an optimization will indeed deterministically output a proportional  allocation if such exists. 

Note that 
minimizing the sum-of-variances over proportional allocations 
{may} 
result in a \emph{set} of {optimal} candidate distributions, rather than in a unique distribution. All our results will hold for \emph{every} distribution in this set. Thus, our results hold regardless of which distribution is selected.

We also note that one can alternatively consider a variant of the 
sum-of-variances 
objective, in which instead of considering the (absolute) value of the allocation to the agents (the value $v_i(A_i)$ for agent $i$), we consider the \emph{fraction} of the value obtained out of the value of the entire set of items (the ratio $v_i(A_i)/v_i(\mathcal{M})$ for agent $i$). That variant might seem more natural to some.  
{Valuations are considered \emph{normalized} if all agents value the set $\mathcal{M}$ the same. 
When valuations are normalized,}
the problem of minimizing the second objective (w.r.t. the fractions) 
is equivalent to the problem of minimizing the 
sum-of-variances (w.r.t. the values), and thus the set of solutions is identical.
Since all our results are for normalized valuations, they extend to the SoVoR objective 
{regardless of whether the valuations are normalized or not.}

We start our investigation by considering agents with an identical additive valuation.
When there are only two agents, then the optimization indeed results in  distributions that all are supported only on allocations that are MMS ex-post (and thus also EFX).
An immediate implication of this result is that finding an 
ex-ante proportional 
sum-of-variances-minimizing distribution is NP‐hard (\cref{prop:np-hard}).

We then move to consider $n\geq 3$ agents with identical additive valuations. 
While for two agents any ex-ante proportional sum‐of‐variances-minimizing distribution is MMS-fair ex-post (\cref{cor:identical2-mms-exact}), we show that such a guarantee does not hold for $n\geq 3$ agents.
That is, we show that in some settings with $n=3$ agents with \emph{identical} additive valuations, 
any ex-ante proportional sum‐of‐variances-minimizing distribution is not MMS-fair ex-post. 
This is particularly disappointing given that an ex-ante proportional distribution that is MMS-fair ex-post {is guaranteed to} 
exist for any setting with $n$ agents with \emph{identical} additive valuations{: take an  MMS allocation of the identical valuation, and give each agent in turn a random bundle from that allocation.}
\footnote{{When valuations are not identical we know that for more than two agents it is impossible to guarantee every agent her MMS \cite{ProcacciaWang14}, and the best we can hope for is constant MMS approximation.}}
Nevertheless, we show that any ex-ante proportional sum‐of‐variances-minimizing distribution is EFX ex-post, which implies {a constant approximation of the MMS ex-post} 
($2/3$ for three agents, and $4/7$ for more agents 
\cite{AmanatidisBirmpasMarkakis2018}). 
Finally, we show that the MMS approximation obtained for three agents is no better than $275 / 304 \approx 90.4\%$. 

We summarize the results for the case of identical valuations and $n\geq 3$ agents: 
\begin{theorem}
[Fairness of Sum‐Of‐Variances Minimizer, identical valuations]\label{thm:intro-identical}
Consider a setting with $n\geq 3$ agents that have the same additive valuation $v$ over $\mathcal{M}$. 
Let $D$ 
be an ex-ante proportional sum‐of‐variances-minimizing distribution. 
Then \(D\) is EFX ex‐post and thus $4/7$-MMS ex-post ({improved to} 
$2/3$-MMS when $n=3$).
Moreover, the MMS approximation for $n=3$ is at most $275 / 304 \approx 90.4\%$.
\end{theorem}

{ 
{The proof of the lower bound is shown in} 
two main steps: 
Initially, we present a structural characterization 
of {all allocations} 
in the support of 
any ex-ante proportional sum-of-variances-minimizing 
distribution. 
{We consider mapping every distribution to a point in $\mathbb{R}_{\geq 0}^n$ that corresponds to the expected values of the $n$ agents, 
and identifying every allocation with the vector 
of values it yields to the agents.
For example, when all agents have valuation $v$, the \emph{ex-ante proportional-share vector} is mapped to the $n$ dimensional vector with all entries being the proportional share $v(\mathcal{M})/n$.}
We show that the sum-of-variances of an 
ex-ante proportional distribution can be expressed as 
the weighted average\footnote{ 
    Where the weight of each allocation is 
    determined by its probability. 
}
of the squared Euclidean distances 
of the value vectors of 
its supporting allocations from the {ex-ante proportional-share vector.}
This is shown to {
imply that every 
ex-ante proportional sum-of-variances-minimizing distribution is 
supported only on allocations {with value vectors} that minimize this 
distance
{(this holds as, given any allocation, one can 
find an ex-ante proportional distribution with every 
allocation in the support having the same distance: 
pick one of the $n$ cyclic permutations of such 
an allocation uniformly at random).} 
Then, building on this characterization, we prove 
that every such distance-minimizing allocation must be 
EFX, and therefore every
ex-ante proportional sum-of-variances
minimizing distribution 
is ex-post EFX. This, in turn, immediately yields a  
constant-factor approximation to the MMS 
\cite{AmanatidisBirmpasMarkakis2018}. 

To show that the MMS approximation is not close to $1$, 
we provide a concrete counterexample 
with three agents. We analyze the counterexample using the 
structural characterization discussed above 
and show that for every distance-minimizing allocation there exists an agent which receives 
a bundle valued only at $275 / 304$ fraction of her $\textnormal{MMS}$, thus the  MMS approximation is at most $275 / 304 \approx 90.4\%$. 
}

We then move to consider agents with {non-identical} 
valuations and prove a strong negative result. {In \cref{prop:nonidentical2-impossibility} we show} 
that once valuations are not identical, picking an ex-ante proportional sum‐of‐variances-minimizing distribution {fails to guarantee} 
ex-post fairness: 
even in the {(simplest)} 
setting of only two agents and two goods, {there exists an instance where} 
the additive valuations are not identical (yet normalized), {and yet} every  ex-ante proportional sum‐of‐variances-minimizing distribution 
places positive probability on an allocation that allocates both goods to the same agent. 
Thus, the supporting ex-post allocation might not be 
EF1 (and thus also not EFX), 
and might not give an agent any fraction of her MMS. 
These negative results hold for a setting with only two agents, although in such settings an MMS allocation 
always exists, and such allocations are also EFX.

We summarize the results for the case of non-identical valuations: 

\begin{theorem}[Failure for Non‐Identical Agents]\label{thm:intro-non-identical}
{When agents may have} 
non-identical additive valuations, there is  an instance for which 
the support of every  ex-ante proportional sum‐of‐variances-minimizing distribution 
includes an allocation which is not EF1 ex-post, and it does not provide any constant-factor 
approximation to the MMS ex-post. 

The claim holds even for instances with only two  agents with additive valuations over two goods and the valuations are normalized (both agents have the same value for the set of all goods).
\end{theorem}

The theorem shows that sum-of-variances minimization 
(subject to ex-ante proportionality) fails to ensure ex-post fairness. 
But one may wonder if selecting another objective function over the agents' variances would solve the problem. 
We show that multiple other natural objective functions similarly fail to ensure ex-post fairness.  
Specifically, {we prove the same negative result when the objective is to minimize the maximal variance of any 
agent}.\footnote{The same result also holds when the 
objective is to minimizing the maximal standard deviation, 
as the optimization problem is equivalent.}
{In addition, we establish that minimizing the sum of 
standard deviations--a natural variant 
of the sum-of-variances objective--also fails to 
guarantee ex-post fairness.}
Finally, we analyze minimizing the ``variance of variances'' and minimizing the ``standard deviation of the standard deviations''. Intuitively, these latter two objectives can be thought of as functions which aspire to keep the variances (standard deviations) 
of the agents as ``concentrated'' as possible.    
To formalize this, we introduce a random variable $X$, which, with uniform probability, takes as its value the variance (standard deviation) of a given agent. The minimization objective is then the variance (standard deviation) of $X$. 
Our analysis reveals that all of the above  minimization objectives inherit the same negative results as the sum-of-variances objective.


\subsection{Additional Related Work}\label{sec:related}
There is extensive research on fair allocation, much more than we can survey in this paper. Below we focus on papers most related to our work. 
The reader is referred to the books 
\cite{brams1996fair,moulin2004fair} and surveys 
\cite{AzizSurvey2022, fairSurvey2022} for general background.

\textbf{Best Of Both Worlds.} 
A recent line of work in fair division focuses on designing 
``Best-of-Both-Worlds (BoBW)'' algorithms, which aim to simultaneously achieve 
strong ex-ante and ex-post fairness guarantees for randomized allocations 
\cite{Aziz2019probabilistic, Freeman2020, Aziz2020, Babaioff2022, Hoefer2022, AzizGG23, FeldmanMNP24,
AleksandrovAGW15, HalpernPPPS20, BabaioffEF21}. 
This research direction was initiated by Aziz~\cite{Aziz2019probabilistic}, 
who proposed a probabilistic framework for combining fairness 
notions across various settings, including voting, allocation, and matching. 
Building on this foundation, Freeman et al.~\cite{Freeman2020} were the 
first to formally coin the term “Best-of-Both-Worlds” in the resource 
allocation context, and they devised a polynomial-time algorithm that outputs a distribution over allocations 
which is envy-free in expectation (ex-ante EF) and each allocation in the support is envy-free up to 
one good (ex-post EF1). 
{Aziz~\cite{Aziz2020} has presented a simpler combinatorial algorithm with the same  guarantees.}
Babaioff et al.~\cite{Babaioff2022}
extend the Best‑of‑Both‑Worlds framework 
to include share‑based ex‑post fairness. 
Specifically, they present a deterministic 
polynomial‑time algorithm that produces a distribution 
over allocations that is ex‑ante proportional  
and ensures that each realized allocation gives 
{each} 
agent at least $1 / 2$ of her maximin share ($1 / 2$-MMS ex-post). 
Other works focus 
on different structured valuation classes,
such as
matroid rank valuations \cite{BabaioffEF21}, 
binary additive 
valuations \cite{AleksandrovAGW15, HalpernPPPS20},
subadditive valuations \cite{FeldmanMNP24}
and multi-demand valuations \cite{Hoefer2022},
or consider the case of arbitrary entitlements 
\cite{AzizGG23, Hoefer2022}.
While all these works aim to explicitly design BoBW algorithms, 
we examine the extent to which BoBW results can be derived from indirect optimization of some natural global functions of the distributions.

\textbf{Relation to Nash Social Welfare (NSW) 
Optimization.} As our study can be viewed as an attempt to 
ensure ex-post fairness via indirect optimization of some 
global function, it is somewhat similar to the approach 
taken when maximizing Nash Social Welfare (NSW) over 
allocations, which is also a global optimization goal. 
For NSW maximization, this approach can be viewed as 
successful, as for agents with additive valuations it 
ensures EF1 ex-post as well as Pareto optimality 
\cite{caragiannis2019unreasonable}. 
Contrary to the success of NSW maximization,
our result shows that when the additive valuations are not 
identical, the approach of minimizing the sum-of-variances 
over {all ex-ante proportional} distributions 
does not provide good fairness ex-post.  

\textbf{MMS Approximation.} 
The maximin share (MMS) was introduced by Budish \cite{Budish11} to handle the indivisibility of goods.  
It has been shown that MMS allocations need not exist even for the case of three 
agents with additive valuations  
\cite{ProcacciaWang14}. 
This led to the study of approximate MMS allocations.
A long line of works delivers 
steadily improving constant-factor guarantees; 
\cite{ProcacciaWang14, BarmanM17, GhodsiHSSY17, Garg2019, 
AmanatidisMNS17, GargT21, AkramiGST23}
with the current best result by \cite{AkramiG24}
proving the existence of $(\frac{3}{4} + \frac{3}{3836})$-MMS 
allocations. 
Regarding the upper bound, 
\cite{FeigeST21} constructed 
an example with $n = 3$  agents and 9 goods
in which no allocation gives every agent more than
$39/40$ of her MMS and for 
$n\geq 4$ their construction 
gives an example in which no allocation is better 
than $(1 - n^{-4})$-MMS. 
{It has been shown \cite{AmanatidisBirmpasMarkakis2018}
that EFX implies $2/3$-MMS for $n\le3$ 
and $4/7$-MMS for $n\ge4$. We 
use this to drive MMS approximation 
results for agents with identical values under 
sum-of-variances minimization.}

Beyond additivity, 
there exist works studying 
{MMS approximations in} different settings,
such as XOS valuations \cite{AkramiMSS23},
unequal entitlements \cite{FarhadiGHLPSSY19}, 
and chores \cite{HuangL21}.

\textbf{EF1/EFX.} 
{
The concept of EF1 was first implicit 
in the work of Lipton et al.~\cite{LiptonMMS04}, 
who proved that an EF1 allocation always exists 
for $n$ agents with arbitrary monotone valuations, 
and was later formally 
introduced by Budish~\cite{Budish11}.

The stronger notion of EFX was introduced by
Caragiannis et al. \cite{caragiannis2019unreasonable}. 
Plaut and Roughgarden
\cite{PlautR20} proved that
for two agents with monotone valuations an EFX allocation always exists.
Chaudhury et al.~\cite{EFX3agents}
showed that for three agents with additive valuations, 
an EFX allocation always exists.
It remains a major open problem 
whether EFX allocations
exist for four or more agents with additive valuations,
although progress has been made in various relaxations,
such as EFX with charity 
\cite{ChaudhuryKMS21}, 
allowing at most one unallocated good 
for four agents
\cite{AlmostEFX4agents}, 
limiting agents to only 
two or three valuations types 
\cite{Mahara23, HVGN025}, 
and others related papers 
\cite{DeligkasEKS24,AmanatidisBFHV21,GargM23,AmanatidisMN20,AmanatidisFS24}.
}

}

\subsection{Paper Organization}
The remainder of the paper is organized as follows.
\cref{sec:preliminaries} introduces notation, standard fairness definitions (MMS, EF1 and EFX), and formalizes the sum‐of‐variances objective.
\cref{sec:identical2} studies two agents with identical additive valuations, showing that ex-ante proportional sum‐of‐variances-minimizing distributions must be supported only on MMS allocations. This implies that computing such a distribution is NP‐hard. \cref{sec:identical3} considers settings with multiple identical agents:
we first show that any allocation in the support must be an allocation
with a value vector of minimal distance to the ex-ante proportional-share 
vector.
We then show that such distributions might not be supported on MMS allocations, yet they are supported only on EFX allocations, which imply MMS approximation results.  
\cref{sec:nonidentical2} presents our main negative result, showing that once valuations are not identical, ex-ante proportional 
sum-of-variances-minimizing distributions 
fundamentally fail to guarantee ex-post fairness: 
the supporting ex-post allocation might not be EF1 (and thus also not EFX), and might not give an agent any constant-fraction of her MMS.
\cref{sec:beyond} discusses variants of the objective function and shows similar negative results.  
\cref{sec:conclusion} concludes. 



%% file: references.bib
@incollection{Aziz2019probabilistic,
  author    = {Haris Aziz},
  title     = {{A Probabilistic Approach to Voting, Allocation, Matching, and Coalition Formation}},
  booktitle = {{The Future of Economic Design: The Continuing Development of a Field as Envisioned by Its Researchers}},
  editor    = {{Jean-Fran{\c{c}}ois Laslier and Herv{\'e} Moulin and M. Remzi Sanver and William S. Zwicker}},
  publisher = {Springer International Publishing},
  address   = {Cham},
  year      = {2019},
  pages     = {45--50}
}

@article{Freeman2020,
  author       = {Rupert Freeman and
                  Nisarg Shah and
                  Rohit Vaish},
  title        = {{Best of Both Worlds: Ex Ante and Ex Post Fairness in Resource Allocation}},
  journal      = {Oper. Res.},
  volume       = {72},
  number       = {4},
  pages        = {1674--1688},
  year         = {2024},
}

@inproceedings{Aziz2020,
  author       = {Haris Aziz},
  title        = {{Simultaneously Achieving Ex-ante and Ex-post Fairness}},
  booktitle    = {{Web and Internet Economics {(WINE)}}},
  series       = {Lecture Notes in Computer Science},
  volume       = {12495},
  pages        = {341--355},
  year         = {2020},
}

@inproceedings{Babaioff2022,
  author       = {Moshe Babaioff and
                  Tomer Ezra and
                  Uriel Feige},
  title        = {{On Best-of-Both-Worlds Fair-Share Allocations}},
  booktitle    = {Web and Internet Economics {(WINE)}},
  volume       = {13778},
  pages        = {237--255},
  year         = {2022},
}

@article{Hoefer2022,
  author       = {Martin Hoefer and
                  Marco Schmalhofer and
                  Giovanna Varricchio},
  title        = {{Best of Both Worlds: Agents with Entitlements}},
  journal      = {J. Artif. Intell. Res.},
  volume       = {80},
  pages        = {559--591},
  year         = {2024},
}

@inproceedings{AzizGG23,
author = {Aziz, Haris and Ganguly, Aditya and Micha, Evi},
title = {{Best of Both Worlds Fairness under Entitlements}},
year = {2023},
booktitle = {{Proceedings of the 2023 International Conference on Autonomous Agents and Multiagent Systems (AAMAS)}},
pages = {941–948},
numpages = {8},
}

@inproceedings{FeldmanMNP24,
author = {Feldman, Michal and Mauras, Simon and Narayan, Vishnu V. and Ponitka, Tomasz},
title = {{Breaking the Envy Cycle: Best-of-Both-Worlds Guarantees for Subadditive Valuations}},
year = {2024},
booktitle = {{Proceedings of the 25th ACM Conference on Economics and Computation (ACM-EC)}},
pages = {1236–1266},
numpages = {31},
}

@inproceedings{AleksandrovAGW15,
  author       = {Martin Aleksandrov and
                  Haris Aziz and
                  Serge Gaspers and
                  Toby Walsh},
  title        = {{Online Fair Division: Analysing a Food Bank Problem}},
  booktitle    = {{Proceedings of the Twenty-Fourth International Joint Conference on
                  Artificial Intelligence {(IJCAI)}}},
  pages        = {2540--2546},
  year         = {2015},
}

@inproceedings{HalpernPPPS20,
  author       = {Daniel Halpern and
                  Ariel D. Procaccia and
                  Alexandros Psomas and
                  Nisarg Shah},
  title        = {{Fair Division with Binary Valuations: One Rule to Rule Them All}},
  booktitle    = {Web and Internet Economics {(WINE)}},
  volume       = {12495},
  pages        = {370--383},
  year         = {2020},
}

@inproceedings{BabaioffEF21,
  author       = {Moshe Babaioff and
                  Tomer Ezra and
                  Uriel Feige},
  title        = {{Fair and Truthful Mechanisms for Dichotomous Valuations}},
  booktitle    = {Thirty-Fifth {AAAI} Conference on Artificial Intelligence {(AAAI)}
                  2021},
  pages        = {5119--5126},
  year         = {2021},
}

@article{Budish11,
  author    = {Budish, Eric},
  title     = {{The Combinatorial Assignment Problem: Approximate Competitive Equilibrium from Equal Incomes}},
  journal   = {Journal of Political Economy},
  volume    = {119},
  number    = {6},
  pages     = {1061--1103},
  year      = {2011},
}

@article{ProcacciaWang14,
  author       = {David Kurokawa and
                  Ariel D. Procaccia and
                  Junxing Wang},
  title        = {{Fair Enough: Guaranteeing Approximate Maximin Shares}},
  journal      = {J. {ACM}},
  volume       = {65},
  number       = {2},
  pages        = {8:1--8:27},
  year         = {2018},
}

@article{AmanatidisMNS17,
  author       = {Georgios Amanatidis and
                  Evangelos Markakis and
                  Afshin Nikzad and
                  Amin Saberi},
  title        = {{Approximation Algorithms for Computing Maximin Share Allocations}},
  journal      = {{ACM} Trans. Algorithms},
  volume       = {13},
  number       = {4},
  pages        = {52:1--52:28},
  year         = {2017},
}

@article{BarmanM17,
  author       = {Siddharth Barman and
                  Sanath Kumar Krishnamurthy},
  title        = {{Approximation Algorithms for Maximin Fair Division}},
  journal      = {{ACM} Trans. Economics and Comput.},
  volume       = {8},
  number       = {1},
  pages        = {5:1--5:28},
  year         = {2020},
}

@inproceedings{GhodsiHSSY17,
  author       = {Mohammad Ghodsi and
                  Mohammad Taghi Hajiaghayi and
                  Masoud Seddighin and
                  Saeed Seddighin and
                  Hadi Yami},
  title        = {{Fair Allocation of Indivisible Goods: Improvements and Generalizations}},
  booktitle    = {{Proceedings of the 2018 {ACM} Conference on Economics and Computation (ACM-EC)}},
  pages        = {539--556},
  year         = {2018},
}

@article{GargT21,
  author       = {Jugal Garg and
                  Setareh Taki},
  title        = {{An Improved Approximation Algorithm for Maximin Shares}},
  journal      = {Artif. Intell.},
  volume       = {300},
  pages        = {103547},
  year         = {2021},
}

@inproceedings{AkramiGST23,
  author       = {Hannaneh Akrami and
                  Jugal Garg and
                  Eklavya Sharma and
                  Setareh Taki},
  title        = {{Simplification and Improvement of {MMS} Approximation}},
  booktitle    = {Proceedings of the Thirty-Second International Joint Conference on
                  Artificial Intelligence {(IJCAI)}},
  pages        = {2485--2493},
  publisher    = {ijcai.org},
  year         = {2023},
}

@inproceedings{AkramiG24,
  author       = {Hannaneh Akrami and
                  Jugal Garg},
  title        = {{Breaking the 3/4 Barrier for Approximate Maximin Share}},
  booktitle    = {Proceedings of the 2024 {ACM-SIAM} Symposium on Discrete Algorithms
                  {(SODA)}},
  pages        = {74--91},
  year         = {2024},
}

@inproceedings{FeigeST21,
  author       = {Uriel Feige and
                  Ariel Sapir and
                  Laliv Tauber},
  title        = {{A Tight Negative Example for {MMS} Fair Allocations}},
  booktitle    = {{Web and Internet Economics {(WINE)}}},
  series       = {Lecture Notes in Computer Science},
  volume       = {13112},
  pages        = {355--372},
  year         = {2021},
}

@inproceedings{AkramiMSS23,
  author       = {Hannaneh Akrami and
                  Kurt Mehlhorn and
                  Masoud Seddighin and
                  Golnoosh Shahkarami},
  title        = {{Randomized and Deterministic Maximin-share Approximations for Fractionally
                  Subadditive Valuations}},
  booktitle    = {{Annual Conference on Neural Information Processing Systems (NeurIPS)}},
  year         = {2023},
}

@article{FarhadiGHLPSSY19,
  author       = {Alireza Farhadi and
                  Mohammad Ghodsi and
                  Mohammad Taghi Hajiaghayi and
                  S{\'{e}}bastien Lahaie and
                  David M. Pennock and
                  Masoud Seddighin and
                  Saeed Seddighin and
                  Hadi Yami},
  title        = {{Fair Allocation of Indivisible Goods to Asymmetric Agents}},
  journal      = {J. Artif. Intell. Res.},
  volume       = {64},
  pages        = {1--20},
  year         = {2019},
}

@inproceedings{AmanatidisBirmpasMarkakis2018,
  author       = {Georgios Amanatidis and
                  Georgios Birmpas and
                  Vangelis Markakis},
  title        = {{Comparing Approximate Relaxations of Envy-Freeness}},
  booktitle    = {Proceedings of the Twenty-Seventh International Joint Conference on                  Artificial Intelligence {(IJCAI)}},
  pages        = {42--48},
  year         = {2018},
}

@inproceedings{HuangL21,
  author       = {Xin Huang and
                  Pinyan Lu},
  title        = {{An Algorithmic Framework for Approximating Maximin Share Allocation
                  of Chores}},
  booktitle    = {{The 22nd {ACM} Conference on Economics and Computation (ACM-EC)}},
  pages        = {630--631},
  year         = {2021},
}

@InProceedings{Garg2019,
  author =	{Garg, Jugal and McGlaughlin, Peter and Taki, Setareh},
  title =	{{Approximating Maximin Share Allocations}},
  booktitle =	{2nd Symposium on Simplicity in Algorithms (SOSA 2019)},
  pages =	{20:1--20:11},
  year =	{2019},
  volume =	{69},
}

@inproceedings{LiptonMMS04,
author = {Lipton, R. J. and Markakis, E. and Mossel, E. and Saberi, A.},
title = {{On approximately fair allocations of indivisible goods}},
year = {2004},
booktitle = {{Proceedings of the 5th ACM Conference on Electronic Commerce (ACM-EC)}},
pages = {125–131},
numpages = {7},
}

@article{caragiannis2019unreasonable,
  author  = {Caragiannis, Ioannis and Kurokawa, David and Moulin, Herv{\'e} and Procaccia, Ariel D. and Shah, Nisarg and Wang, Junxing},
  title   = {{The Unreasonable Fairness of Maximum Nash Welfare}},
  journal = {ACM Transactions on Economics and Computation},
  volume  = {7},
  number  = {3},
  pages   = {1--32},
  year    = {2019}
}

@article{PlautR20,
  author       = {Benjamin Plaut and
                  Tim Roughgarden},
  title        = {{Almost Envy-Freeness with General Valuations}},
  journal      = {{SIAM} J. Discret. Math.},
  volume       = {34},
  number       = {2},
  pages        = {1039--1068},
  year         = {2020},
}

@article{EFX3agents,
  author = {Chaudhury, Bhaskar Ray and Garg, Jugal and Mehlhorn, Kurt},
  title = {{EFX Exists for Three Agents}},
  year = {2024},
  volume = {71},
  number = {1},
  journal = {J. ACM},
  month = feb,
  articleno = {4},
  numpages = {27},
}

@article{ChaudhuryKMS21,
  author       = {Bhaskar Ray Chaudhury and
                  Telikepalli Kavitha and
                  Kurt Mehlhorn and
                  Alkmini Sgouritsa},
  title        = {{A Little Charity Guarantees Almost Envy-Freeness}},
  journal      = {{SIAM} J. Comput.},
  volume       = {50},
  number       = {4},
  pages        = {1336--1358},
  year         = {2021},
}

@article{AlmostEFX4agents,
  author       = {Ben Berger and
                  Avi Cohen and
                  Michal Feldman and
                  Amos Fiat},
  title        = {{(Almost Full) {EFX} Exists for Four Agents (and Beyond)}},
  journal      = {CoRR},
  volume       = {abs/2102.10654},
  year         = {2021},
}

@article{Mahara23,
  author       = {Ryoga Mahara},
  title        = {{Existence of {EFX} for two additive valuations}},
  journal      = {Discret. Appl. Math.},
  volume       = {340},
  pages        = {115--122},
  year         = {2023},
}

@inproceedings{HVGN025,
  author       = {Vishwa Prakash HV and
                  Pratik Ghosal and
                  Prajakta Nimbhorkar and
                  Nithin Varma},
  title        = {{{EFX} Exists for Three Types of Agents}},
  booktitle    = {{Proceedings of the 26th {ACM} Conference on Economics and Computation (ACM-EC)}},
  pages        = {101--128},
  year         = {2025},
}

@inproceedings{DeligkasEKS24,
  author       = {Argyrios Deligkas and
                  Eduard Eiben and
                  Viktoriia Korchemna and
                  Simon Schierreich},
  title        = {{The Complexity of Fair Division of Indivisible Items with Externalities}},
  booktitle    = {{Thirty-Eighth {AAAI} Conference on Artificial Intelligence {(AAAI)}}},
  pages        = {9653--9661},
  year         = {2024},
}

@article{AmanatidisBFHV21,
  author       = {Georgios Amanatidis and
                  Georgios Birmpas and
                  Aris Filos{-}Ratsikas and
                  Alexandros Hollender and
                  Alexandros A. Voudouris},
  title        = {{Maximum Nash welfare and other stories about {EFX}}},
  journal      = {Theor. Comput. Sci.},
  volume       = {863},
  pages        = {69--85},
  year         = {2021},
}

@article{GargM23,
  author       = {Jugal Garg and
                  Aniket Murhekar},
  title        = {{Computing fair and efficient allocations with few utility values}},
  journal      = {Theor. Comput. Sci.},
  volume       = {962},
  pages        = {113932},
  year         = {2023},
}

@article{AmanatidisMN20,
  author       = {Georgios Amanatidis and
                  Evangelos Markakis and
                  Apostolos Ntokos},
  title        = {{Multiple Birds with One Stone: Beating 1/2 for EFX and GMMS via Envy Cycle Elimination}},
  journal      = {Theor. Comput. Sci.},
  volume       = {841},
  pages        = {94--109},
  year         = {2020},
}

@inproceedings{AmanatidisFS24,
  author       = {Georgios Amanatidis and
                  Aris Filos{-}Ratsikas and
                  Alkmini Sgouritsa},
  title        = {{Pushing the Frontier on Approximate {EFX} Allocations}},
  booktitle    = {{Proceedings of the 25th {ACM} Conference on Economics and Computation (ACM-EC)}},
  pages        = {1268--1286},
  year         = {2024},
}

@inproceedings{Korf10,
  author       = {Richard E. Korf},
  title        = {{Objective Functions for Multi-Way Number Partitioning}},
  booktitle    = {{Proceedings of the Third Annual Symposium on Combinatorial Search {(SOCS)}}},
  pages        = {71--72},
  year         = {2010},
}

@book{Rockafellar1970,
  author    = {Rockafellar, R. Tyrrell},
  title     = {{Convex Analysis}},
  publisher = {Princeton University Press},
  address   = {Princeton, NJ},
  year      = {1970},
}

@book{brams1996fair,
	title={{Fair Division: From cake-cutting to dispute resolution}},
	author={Brams, Steven J and Taylor, Alan D},
	year={1996},
	publisher={Cambridge University Press}
}

@book{moulin2004fair,
	title={{Fair Division and Collective Welfare}},
	author={Moulin, Herv{\'e}},
	year={2004},
	publisher={MIT press}
}

@article{AzizSurvey2022,
author = {Aziz, Haris and Li, Bo and Moulin, Herv\'{e} and Wu, Xiaowei},
title = {{Algorithmic Fair Allocation of Indivisible Items: A Survey and New Questions}},
year = {2022},
volume = {20},
number = {1},
journal = {SIGecom Exch.},
pages = {24–40},
numpages = {17},
}

@inproceedings{fairSurvey2022,
title = {{Fair Division of Indivisible Goods: A Survey}},
author = "Georgios Amanatidis and Georgios Birmpas and Aris Filos-Ratsikas and Voudouris, {Alexandros A.}",
year = "2022",
month = jul,
day = "29",
pages = "5385--5393",
booktitle = "The International Joint Conference on Artificial Intelligence - Survey Track (IJCAI)",
}
